\documentclass[10pt]{article}

\usepackage{natbib}
\usepackage{amsthm}

\newtheorem{proposition}{Proposition}

\usepackage[utf8]{inputenc} 
\usepackage[T1]{fontenc}   
\usepackage{hyperref}     
\usepackage{url}           
\usepackage{booktabs}    
\usepackage{amsfonts}    
\usepackage{nicefrac}     
\usepackage{microtype}   

\usepackage{amsmath}
\usepackage{nccmath}
\usepackage{amssymb}
\usepackage{mathrsfs}
\usepackage{graphicx}
\usepackage{ulem}
\usepackage{algorithmic}
\usepackage{algorithm}

\usepackage{array}
\usepackage{xcolor}

\usepackage{multirow}
\usepackage{array, makecell, rotating}

\DeclareMathOperator{\E}{\mathbb{E}}
\DeclareMathOperator*{\argmax}{arg\,max}
\DeclareMathOperator*{\argmin}{arg\,min}
\DeclareMathOperator*{\minimize}{minimize}

\begin{document}

\title{Robust Bayesian Model Selection for Variable Clustering with the Gaussian Graphical Model}

\author{
  Daniel Andrade\footnote{The first author is also affiliated with NEC Corporation.} \\
  \footnotesize Graduate University of Advanced Studies (SOKENDAI) \\
     \footnotesize  10-3 Midoricho, Tachikawa, Tokyo, 190-8562, Japan \\
   \footnotesize  andrade@ism.ac.jp
  \and 
  Akiko Takeda  \\
  \footnotesize   Department of Creative Informatics \\
\footnotesize   The University of Tokyo \\
\footnotesize   7-3-1 Hongo, Bunkyo-ku, Tokyo, 113-8656, Japan \\
\footnotesize   takeda@mist.i.u-tokyo.ac.jp 
  \and 
  Kenji Fukumizu  \\
  \footnotesize    The Institute of Statistical Mathematics \\
  \footnotesize  10-3 Midoricho, Tachikawa, Tokyo, 190-8562, Japan \\
  \footnotesize    fukumizu@ism.ac.jp
}

\maketitle

\begin{abstract}
Variable clustering is important for explanatory analysis. However, only few dedicated methods for variable clustering with the Gaussian graphical model have been proposed.
Even more severe, small insignificant partial correlations due to noise can dramatically change the clustering result when evaluating for example with the Bayesian Information Criteria (BIC).
In this work, we try to address this issue by proposing a Bayesian model that accounts for negligible small, but not necessarily zero, partial correlations. 
Based on our model, we propose to evaluate a variable clustering result using the marginal likelihood. 
To address the intractable calculation of the marginal likelihood, we propose two solutions: one based on a variational approximation, and another based on MCMC.
Experiments on simulated data shows that the proposed method is similarly accurate as BIC in the no noise setting, but considerably more accurate when there are noisy partial correlations.
Furthermore, on real data the proposed method provides clustering results that are intuitively sensible, which is not always the case when using BIC or its extensions. 
\end{abstract}

\section{Introduction}

The Gaussian graphical model (GGM) has become an invaluable tool for detecting partial correlations between variables.
Assuming the variables are jointly drawn from a multivariate normal distribution, the sparsity pattern of the precision matrix reveals which pairs of variables are independent given all other variables \citep{anderson2004introduction}.
In particular, we can find clusters of variables that are mutually independent, by grouping the variables according their entries in the precision matrix.

However, in practice, it can be difficult to find a meaningful clustering due to the noise of the entries in the partial correlations.
The noise can be due to the sampling, this is in particular the case when $n$ the number of observations is small, or due to small non-zero partial correlations in the true precision matrix that might be considered as insignificant.
Here in this work, we are particularly interested in the latter type of noise.
In the extreme, small partial correlations might lead to a connected graph of variables, where no grouping of variables can be identified.
For an exploratory analysis such a result might not be desirable. 

As an alternative, we propose to find a clustering of variables, such that the partial correlation between two variables in different groups is negligibly small, but not necessarily zero.
The open question, which we try to address here, is whether there is a principled model selection criteria for this scenario.

For example, the Bayesian Information Criteria (BIC) \citep{schwarz1978estimating} is a popular model selection criteria for the Gaussian graphical model. However, in the noise setting it does not have any formal guarantees.
As a solution, we propose here a Bayesian model that explicitly accounts for small partial correlations between variables in different clusters.

Under our proposed model, the marginal likelihood of the data can then be used to identify the correct (if there is a ground truth in theory), or at least a meaningful clustering (in practice) that helps analysis. 
The marginal likelihood of our model does not have an analytic solution. Therefore, we provide two approximations.
The first is a variational approximation, the second is based on MCMC.

Experiments on simulated data show that the proposed method is similarly accurate as BIC in the no noise setting, but considerably more accurate when there are noisy partial correlations.
The proposed method also compares favorable to two previously proposed methods for variable clustering and model selection, namely the Clustered Graphical Lasso (CGL) \citep{tan2015cluster} and the Dirichlet Process Variable Clustering (DPVC) \citep{palla2012nonparametric} method.

Our paper is organized as follows. In Section \ref{sec:relatedWork}, we discuss previous works related to variable clustering and model selection. 
In Section \ref{sec:basicBayesianModel}, we introduce a basic Bayesian model for evaluating variable clusterings, which we then extend in Section \ref{sec:proposedModel} to handle noise on the precision matrix. 
For the proposed model, which can handle noise in the precision matrix, the calculation of the marginal likelihood is infeasible and we describe our approximation strategy in Section \ref{sec:estimationMarginalLikelihood}.
Since enumerating all possible clusterings is intractable, we describe in Section \ref{sec:restrictingHypothesesSpace} an heuristic based on spectral clustering to limit the number of candidate clusterings.
We evaluate the proposed method on synthetic and real data in Sections \ref{sec:simulatedData} and \ref{sec:realDataExperiments}, respectively. Finally, we discuss our findings in Section \ref{sec:conclusions}.

\section{Related Work} \label{sec:relatedWork}
Finding a clustering of variables is equivalent to finding an appropriate block structure of the covariance matrix.
Recently, \cite{tan2015cluster} and \cite{devijver2016block} suggested to detect block diagonal structure by thresholding the absolute values of the covariance matrix. Their methods perform model selection using the mean squared error of randomly left out elements of the covariance matrix \citep{tan2015cluster}, and a slope heuristic \citep{devijver2016block}. 

Also several Bayesian latent variable models have been proposed for this task \citep{marlin2009sparse,sun2014adaptive,palla2012nonparametric}.
Each clustering, including the number of clusters, is either evaluated using the variational lower bound \citep{marlin2009sparse}, or by placing a Dirichlet Process prior over clusterings \citep{palla2012nonparametric,sun2014adaptive}.
However, all of the above methods assume that the partial correlations of variables across clusters are exactly zero.

An exception is the work in \citep{marlin2009group} which proposes to regularize the precision matrix such that 
partial correlations of variables that belong to the same cluster are penalized less than those belonging to different clusters.
For that purpose they introduce three hyper-parameters, $\lambda_1$ (for within cluster penalty), $\lambda_0$ (for across clusters), with $\lambda_0 > \lambda_1$, and $\lambda_D$ for a penalty of the diagonal elements.
The clusters do not need to be known a-priori and are estimated by optimizing a lower bound on the marginal likelihood. 
As such their method can also find variable clusterings, even when the true partial correlation of variables in different clusters  is not exactly zero. 
However, the clustering result is influenced by three hyperparameters $\lambda_0, \lambda_1$, and $\lambda_D$ which have to be determined using cross-validation.

Recently, the work in \citep{sun2015inferring,hosseini2016learning} relaxes the assumption of a clean block structure by allowing some variables to correspond to two clusters.
The model selection issue, in particular, determining the number of clusters, is either addressed with some heuristics \citep{sun2015inferring} or cross-validation \citep{hosseini2016learning}. 

\section{The Bayesian Gaussian Graphical Model for Clustering} \label{sec:basicBayesianModel}

Our starting point for variable clustering is the following Bayesian Gaussian graphical model.
Let us denote by $p$ the number of variables, and $n$ the number of observations.
We assume that each observation $\mathbf{x} \in \mathbb{R}^p$ is generated i.i.d. from a multivariate normal distribution with zero mean and covariance matrix $\Sigma$.
Assuming that there are $k$ groups of variables that are mutually independent, we know that, after appropriate permutation of the variables, $\Sigma$ has the following block structure

\begin{align*}
\Sigma = 
\left(
\begin{array}{ccc}
\Sigma_1 & 0 & 0 \\
0 & \ddots  & 0  \\
0 & 0  & \Sigma_k
\end{array} \right) \, ,
\end{align*}
where $\Sigma_j \in \mathbb{R}^{p_j \times p_j}$, and $p_j$ is the number of variables in cluster $j$.

By placing an inverse Wishart prior over each block $\Sigma_j$, we arrive at the following Bayesian model
\begin{equation}
\begin{aligned}  \label{eq:basicModel}
& p(\mathbf{x}_1, ..., \mathbf{x}_n, \Sigma | \{ \nu_{j} \}_j, \{\Sigma_{j,0}\}_j, \mathcal{C}) \\  
& = \prod_{i = 1}^n \text{Normal} (\mathbf{x}_i  | \mathbf{0}, \Sigma)  \prod_{j=1}^k \text{InvW}(\Sigma_j | \nu_{j}, \Sigma_{j,0})   \, ,
\end{aligned}
\end{equation}
where $\nu_{j}$ and $\Sigma_{j,0}$, are the degrees of freedom and the scale matrix, respectively. We set
$\nu_{j} =p_j +1, \Sigma_j = I_{p_j}$ leading to a non-informative prior on $\Sigma_j$. 
$\mathcal{C}$ denotes the variable clustering which imposes the block structure on $\Sigma$. 
We will refer to this model as the basic inverse Wishart prior model.

Assuming we are given a set of possible variable clusterings $\mathscr{C}$, we can then choose the clustering $\mathcal{C}^*$ that maximizes the posterior probability of the clustering, i.e.
\begin{align*}
\mathcal{C}^* = \argmax_{\mathcal{C} \in \mathscr{C}} p(\mathcal{C} | \mathscr{X})  = \argmax_{\mathcal{C} \in \mathscr{C}} p(\mathscr{X} | \mathcal{C}) \cdot p(\mathcal{C}) \, ,
\end{align*}
where we denote by $\mathscr{X}$ the observations $\mathbf{x}_1, ..., \mathbf{x}_n$, and $p(\mathcal{C})$ is a prior over the clusterings which we assume to be uniform.
Here, we refer to $p(\mathscr{X} | \mathcal{C})$ as the marginal likelihood (given the clustering). For the basic inverse Wishart prior model the marginal likelihood can be calculated analytically, see e.g. \citep{lenkoski2011computational}.

\section{Proposed Model} \label{sec:proposedModel}

In this section, we extend the Bayesian model from Equation \eqref{eq:basicModel} to account for non-zero partial correlations between variables in different clusters.
For that purpose we introduce the matrix $\Sigma_{\epsilon} \in \mathbb{R}^{p \times p}$ that models the noise on the precision matrix.
The full joint probability of our model is given as follows:

\begin{equation}
\begin{aligned}  \label{eq:proposedModel}
&p(\mathbf{x}_1, ..., \mathbf{x}_n, \Sigma, \Sigma_{\epsilon} | \nu_{\epsilon}, \Sigma_{\epsilon,0}, \{ \nu_{j} \}_j, \{\Sigma_{j,0}\}_j, \mathcal{C})  \\
&= \prod_{i = 1}^n \text{Normal} (\mathbf{x}_i  | \mathbf{0}, \Xi) \\
&\quad \cdot \text{InvW}(\Sigma_{\epsilon} | \nu_{\epsilon}, \Sigma_{\epsilon,0}) \prod_{j=1}^k \text{InvW}(\Sigma_j | \nu_{j}, \Sigma_{j,0})   \, ,
\end{aligned}
\end{equation}
where $\Xi := (\Sigma^{-1} + \beta \Sigma_{\epsilon}^{-1})^{-1}$, and 
\begin{align*}
\Sigma := 
\left(
\begin{array}{ccc}
\Sigma_1 & 0 & 0 \\
0 & \ddots  & 0  \\
0 & 0  & \Sigma_k
\end{array} \right) \, . 
\end{align*}
As before, the block structure of $\Sigma$ is given by the clustering $\mathcal{C}$.
The proposed model is the same model as in Equation \eqref{eq:basicModel}, with the main difference that the noise term $\beta \Sigma_{\epsilon}^{-1}$ is added to the precision matrix of the normal distribution.

$1 \gg \beta > 0$ is a hyper-parameter that is fixed to a small positive value accounting for the degree of noise on the precision matrix.
Furthermore, we assume non-informative priors on $\Sigma_j$ and $\Sigma_{\epsilon}$ by setting $\nu_{j} =p_j +1, \Sigma_j = I_{p_j}$ and  $\nu_{\epsilon} = p + 1, \Sigma_{\epsilon,0} = I_p$.

\paragraph{Remark on the parameterization} We note that as an alternative parameterization, we could have defined $\Xi := (\Sigma^{-1} + \Sigma_{\epsilon}^{-1})^{-1}$, and 
instead place a prior on $\Sigma_{\epsilon}$ that encourages $\Sigma_{\epsilon}^{-1}$ to be small in terms of some matrix norm. For example, we could have set
$\Sigma_{\epsilon,0} = \frac{1}{\beta} I_p$.

\section{Estimation of the Marginal Likelihood} \label{sec:estimationMarginalLikelihood}

The marginal likelihood of the data given our proposed model can be expressed as follows:
\begin{align*}
&p(\mathbf{x}_1, ..., \mathbf{x}_n | \nu_{\epsilon}, \Sigma_{\epsilon,0}, \{ \nu_{j} \}_j, \{\Sigma_{j,0}\}_j, \mathcal{C}) \\
&= \int  \text{Normal} (\mathbf{x}_1, ..., \mathbf{x}_n  | \mathbf{0}, \Xi ) \\
&\quad \cdot \prod_{j=1}^k \text{InvW}(\Sigma_j | \nu_{j}, \Sigma_{j,0}) d(\Sigma_{j} \succ 0) \\
&\quad \cdot \text{InvW}(\Sigma_{\epsilon} | \nu_{\epsilon}, \Sigma_{\epsilon,0}) d(\Sigma_{\epsilon} \succ 0) \, .
\end{align*}
where $\Xi := (\Sigma^{-1} + \beta \Sigma_{\epsilon}^{-1})^{-1}$.

Clearly, if $\beta = 0$, we recover the basic inverse Wishart prior model, as discussed in Section \ref{sec:basicBayesianModel}, and the marginal likelihood has a closed form solution due to the conjugacy of the covariance matrix of the Gaussian and the inverse Wishart prior.
However, if $\beta > 0$, there is no analytic solution anymore. Therefore, we propose to either use an estimate based on a variational approximation (Section \ref{sec:variationalApproximation}) or on MCMC (Section \ref{sec:MCMCEstimate}).
Both of our estimates require the calculation of the maximum a posterior solution which we explain first in Section \ref{sec:MAPsolution}.

\paragraph{Remark on BIC type approximation of the marginal likelihood} We note that for our proposed model an approximation of the marginal likelihood using BIC is not sensible.
To see this, recall that BIC consists of two terms: the data log-likelihood under the model with the maximum likelihood estimate, and a penalty depending on the number of free parameters.
The maximum likelihood estimate is 
\begin{align*}
\hat{\Sigma}, \hat{\Sigma}_{\epsilon} = \argmax_{\Sigma, \Sigma_{\epsilon}} \sum_{i = 1}^n \log \text{Normal} (\mathbf{x}_i  | \mathbf{0}, (\Sigma^{-1} + \beta \Sigma_{\epsilon}^{-1})^{-1}) \, ,
\end{align*}
where $S$ is the sample covariance matrix. 
Note that without the specification of a prior, it is valid that $\hat{\Sigma}, \hat{\Sigma}_{\epsilon}$ are not positive definite as long as the matrix $\hat{\Sigma}^{-1} + \beta \hat{\Sigma}_{\epsilon}^{-1}$ is positive definite.
Therefore $\hat{\Sigma}^{-1} + \beta \hat{\Sigma}_{\epsilon}^{-1} = S^{-1}$, and the data likelihood under the model with the maximum likelihood estimate is simply $\sum_{i = 1}^n \log \text{Normal} (\mathbf{x}_i  | \mathbf{0}, S)$, which is independent of the clustering.
The number of \emph{free} parameters is $(p^2 - p) / 2$ which is also independent of the clustering. 
That means, for any clustering we end up with the same BIC.

Furthermore, a Laplacian approximation as used in the generalized Bayesian information criterion \citep{konishi2004bayesian} is also not suitable, since in our case the parameter space is over the positive definite matrices.

\subsection{Calculation of maximum a posterior solution} \label{sec:MAPsolution}

First note that 
\begin{align*}
&p(\Sigma, \Sigma_{\epsilon} | \mathbf{x}_1, ..., \mathbf{x}_n, \nu_{\epsilon}, \Sigma_{\epsilon,0}, \{ \nu_{j} \}_j, \{\Sigma_{j,0}\}_j, \mathcal{C}) \\
&\propto  \text{Normal} (\mathbf{x}_1, ..., \mathbf{x}_n  | \mathbf{0}, \Xi) \\
&\quad \cdot \prod_{j=1}^k \text{InvW}(\Sigma_j | \nu_{j}, \Sigma_{j,0}) \\
&\quad \cdot \text{InvW}(\Sigma_{\epsilon} | \nu_{\epsilon}, \Sigma_{\epsilon,0}) 
\end{align*}
where $\Xi := (\Sigma^{-1} + \beta \Sigma_{\epsilon}^{-1})^{-1}$.

Therefore, 
\begin{align*}
& \log p(\Sigma, \Sigma_{\epsilon} | \mathbf{x}_1, ..., \mathbf{x}_n, \nu_{\epsilon}, \Sigma_{\epsilon,0}, \{ \nu_{j} \}_j, \{\Sigma_{j,0}\}_j, \mathcal{C}) = \\
& -\frac{n}{2} \log |\Xi| -\frac{n}{2} trace (S \Xi^{-1}  ) \\
& -\frac{\nu_{\epsilon} + p + 1}{2} \log |\Sigma_{\epsilon}| - \frac{1}{2} trace (\Sigma_{\epsilon,0} \Sigma_{\epsilon}^{-1}) \\
& + \sum_{j = 1}^{k}  \Big( -\frac{\nu_{j} + p_j + 1}{2} \log |\Sigma_{j}| - \frac{1}{2} trace (\Sigma_{j,0} \Sigma_{j}^{-1}) \Big) \\
& + const \\
=& \, \frac{1}{2} \Big(n \cdot \log |\Xi^{-1}| - n \cdot trace (S \Xi^{-1}  ) \\
& + (\nu_{\epsilon} + p + 1) \cdot \log |\Sigma_{\epsilon}^{-1}| - trace (\Sigma_{\epsilon,0} \Sigma_{\epsilon}^{-1}) \\
& + \sum_{j = 1}^{k} \Big( (\nu_{j} + p_j + 1) \cdot \log |\Sigma_{j}^{-1}| - trace (\Sigma_{j,0} \Sigma_{j}^{-1}) \Big) \Big) \\
&+ const \, , 
\end{align*}
where the constant is with respect to $\Sigma_{\epsilon}, \Sigma_{1}, \ldots \Sigma_{k}$, and $p_j$ denotes the number of variables in cluster $j$.

\paragraph{Solution using a 3-Block ADMM}

Finding the MAP can be formulated as a convex optimization problem by a change of parameterization: 
by defining $X := \Sigma^{-1}$, $X_j := \Sigma^{-1}_j$, and $X_{\epsilon} := \Sigma_{\epsilon}^{-1}$, we get the following convex optimization problem: 
\begin{equation}
\begin{aligned} \label{eq:MAP_opt_problem}
& \minimize_{X \succ 0, X_{\epsilon} \succ 0} \; n \cdot trace(S (X + \beta X_{\epsilon}))  - n \cdot \log |X + \beta X_{\epsilon}| \\
&\quad  + trace(A_{\epsilon} X_{\epsilon})  - a_{\epsilon} \cdot \log |X_{\epsilon}| \\
&\quad + \sum_{j = 1}^{k}  \Big( trace(A_j X_j)  - a_j \cdot \log |X_j| \Big) \, ,
\end{aligned}
\end{equation}
where, for simplifying notation, we introduced the following constants:
\begin{align*}
& A_{\epsilon} := \Sigma_{\epsilon,0} \, , \\
& a_{\epsilon} := \nu_{\epsilon} + p + 1 \, , \\
& A_{j} := \Sigma_{j,0} \, , \\
& a_{j} := \nu_{j} + p_j + 1 \, .
\end{align*}
From this form, we see immediately that the problem is strictly convex jointly in $X_{\epsilon}$ and $X$.\footnote{Since $- log |X|$ is a strictly convex function and $trace(XS)$ is a linear function.}

We further reformulate the problem by introducing an additional variable Z: 
\begin{align*}
& \text{minimize} \; f(X_{\epsilon}, X_1, \ldots, X_k, Z)  \\
& \text{subject to} \; \\
& \; Z = X + \beta X_{\epsilon} \, , \\
& X_{\epsilon}, X_1, \ldots, X_k, Z \succeq 0 \, ,
\end{align*}
with
\begin{align*}
f(X_{\epsilon}, X_1, \ldots, X_k, Z) := \, & n \cdot trace(S Z)  - n \cdot \log |Z| \\
 & + trace(A_{\epsilon} X_{\epsilon}) - a_{\epsilon} \cdot \log |X_{\epsilon}| \\
 & + \sum_{j = 1}^{k}  \Big( trace(A_j X_j)  - a_j \cdot \log |X_j| \Big) \, .
\end{align*}

It is tempting to use a 2-Block ADMM algorithm, like e.g. in \citep{boyd2011distributed}, which leads to two optimization problems: update of $X, X_{\epsilon}$ and update of $Z$.
However, unfortunately, in our case the resulting optimization problem for updating $X, X_{\epsilon}$ does not have an analytic solution. 
Therefore, instead, we suggest the use of a 3-Block ADMM, which updates the following sequence:

\begin{align*}
X^{t + 1}  := & \argmin_{X_1, \ldots, X_k \succ 0} \;  \sum_{j = 1}^{k}  \Big ( trace(A_j X_j)  - a_j \cdot \log |X_j| \Big) \\
&\quad +  trace(U^t  (X + \beta X_{\epsilon}^t - Z^t)) \\
&\quad + \frac{\rho}{2} || X + \beta X_{\epsilon}^t - Z^t ||_F^2 \, , \\
 X_{\epsilon}^{t + 1}  := & \argmin_{X_{\epsilon} \succ 0} \; trace(A_{\epsilon} X_{\epsilon})  - a_{\epsilon} \cdot \log |X_{\epsilon}| \\
 &\quad + trace(U^t  (X^{t+1} + \beta X_{\epsilon} - Z^t))  \\
 &\quad + \frac{\rho}{2} || X^{t+1} + \beta X_{\epsilon} - Z^t ||_F^2 \, , \\
Z^{t+1} := & \argmin_{Z \succ 0} \; n \cdot trace(S Z)  - n \cdot \log |Z| \\ 
&\quad + trace(U^t  (X^{t+1} + \beta X_{\epsilon}^{t+1} - Z))  \\
&\quad +  \frac{\rho}{2} || X^{t+1} + \beta X_{\epsilon}^{t + 1} - Z ||_F^2 \, , \\
U^{t+1} := &  \rho (X^{t+1} + \beta X_{\epsilon}^{t+1} - Z^{t + 1}) + U^t \, ,
\end{align*}

where $U$ is the Lagrange multiplier, and $X^t, Z^t, U^t$, denotes $X, Z, U$ at iteration $t$; $\rho > 0$ is the learning rate.\footnote{In our experiments, we set the learning rate $\rho$ initially to 1.0, and increase it every 100 iterations by a factor of $1.1$. We found experimentally that this speeds-up the convergence of ADMM. }

Each of the above sub-optimization problem can be solved efficiently via the following strategy. 
The zero gradient condition for the first optimization problem with variable $X$ is
\begin{align*}
&  - X_j^{-1} + \frac{\rho}{a_j} X_j = - \frac{1}{a_j} (A_j + U_j + \rho (\beta X_{\epsilon,j} - Z_j)) \, .
\end{align*}

The zero gradient condition for the second optimization problem with variable $X_{\epsilon}$ is
\begin{align*}
&  - X_{\epsilon}^{-1} + \frac{\rho \beta^2}{a_{\epsilon}} X_{\epsilon} = - \frac{1}{a_{\epsilon}} ( A_{\epsilon}  + \beta U + \rho \beta (X - Z)) \, .
\end{align*}

The zero gradient condition for the third optimization problem with variable $Z$ is
\begin{align*}
&  - Z^{-1} + \frac{\rho}{n} Z = \frac{1}{n} ( U  - nS + \rho (X + \beta X_{\epsilon})) \, .
\end{align*}

Each of the above three optimization problem can be solved via an eigenvalue decomposition as follows.
We need to solve $V$ such that it satisfies:
\begin{align*}
- V^{-1} + \lambda V  = R \; \;  \land \; \;  V \succeq 0
\end{align*}
Since $R$ is a symmetric matrix (not necessarily positive or negative semi-definite), we have the eigenvalue decomposition:
\begin{align*}
QLQ^T = R \, ,
\end{align*}
where $Q$ is an orthonormal matrix and $L$ is a diagonal matrix with real values.
Denoting $Y := Q^T V Q$, we have 
\begin{align} \label{eq:diagonalTransformed}
- Y^{-1} + \lambda Y  = L \, ,
\end{align}
Since the solution $Y$ must also be a diagonal matrix, we have $Y_{ij} = 0$, for $j \neq i$, and we must have that 
\begin{align} \label{eq:yTranformed}
- (Y_{ii})^{-1} + \lambda Y_{ii}  = L_{ii} \, .
\end{align}
Then, Equation \eqref{eq:yTranformed} is equivalent to
\begin{align*}
 \lambda Y_{ii}^2  - L_{ii} Y_{ii}  -1 = 0 \, ,
\end{align*}
and therefore one solution is 
\begin{align*}
Y_{ii} = \frac{L_{ii} + \sqrt{L_{ii}^2 + 4 \lambda }}{2 \lambda}  \, .
\end{align*}
Note that for $\lambda > 0$, we have that $Y_{ii} > 0$.  Therefore, we have that the resulting $Y$ solves Equation \eqref{eq:diagonalTransformed} and moreover 
\begin{align*}
V = Q Y Q^T \succ 0 \, .
\end{align*}
That means, we can solve the semi-definite problem with only one eigenvalue decomposition, and therefore is in $O(p^3)$.

Finally, we note that in contrast to the 2-block ADMM, a general 3-block ADMM does not have a convergence guarantee for any $\rho > 0$.
However, using a recent result from \citep{lin2015global}, we can show in Appendix \ref{app:convergence3BlockADMM} that in our case the conditions for convergence are met for any $\rho > 0$.

\subsection{Variational Approximation of the Marginal Likelihood} \label{sec:variationalApproximation}

Here we explain our strategy for the calculation of a variational approximation of the marginal likelihood.
For simplicity, let $\boldsymbol{\theta}$ denote the vector of all parameters, $\mathscr{X}$ the observed data, and $\boldsymbol{\eta}$ the vector of all hyper-parameters.

Let $\hat{\boldsymbol{\theta}}$ denote the posterior mode. 
Furthermore, let $g(\boldsymbol{\theta})$ be an approximation of the posterior distribution $p(\boldsymbol{\theta} | \mathscr{X}, \boldsymbol{\eta}, \mathcal{C})$ that is accurate around the mode $\hat{\boldsymbol{\theta}}$.

Then we have 
\begin{equation}
\begin{aligned} \label{eq:marginalApproxAroundMode}
p(\mathscr{X} | \boldsymbol{\eta}, \mathcal{C}) &= \frac{p(\boldsymbol{\theta} , \mathscr{X} | \boldsymbol{\eta}, \mathcal{C})}{p(\boldsymbol{\theta} | \mathscr{X}, \boldsymbol{\eta}, \mathcal{C})} \\
&=  \frac{p(\hat{\boldsymbol{\theta}}, \mathscr{X} | \boldsymbol{\eta}, \mathcal{C})}{p(\hat{\boldsymbol{\theta}} | \mathscr{X}, \boldsymbol{\eta}, \mathcal{C})} 
\approx \frac{p(\hat{\boldsymbol{\theta}} , \mathscr{X} | \boldsymbol{\eta}, \mathcal{C})}{g(\hat{\boldsymbol{\theta}})} \, .
\end{aligned}
\end{equation}

Note that for the Laplace approximation we would use $g(\boldsymbol{\theta}) = N(\boldsymbol{\theta} | \hat{\boldsymbol{\theta}}, V)$, where $V$ is an appropriate covariance matrix. However, here the posterior $p(\boldsymbol{\theta} | \mathscr{X}, \boldsymbol{\eta}, \mathcal{C})$ is a probability measure over the positive definite matrices and not over $\mathbb{R}^d$, which makes the Laplace approximation inappropriate. 

Instead, we suggest to approximate the posterior distribution \\
$p(\Sigma_{\epsilon}, \Sigma_{1}, \ldots \Sigma_{k} | \mathbf{x}_1, ..., \mathbf{x}_n, \nu_{\epsilon}, \Sigma_{\epsilon,0}, \{ \nu_{j} \}_j, \{\Sigma_{j,0}\}_j, \mathcal{C})$ by the factorized distribution
\begin{align*}
g :=  g_{\epsilon}(\Sigma_{\epsilon}) \cdot \prod_{j=1}^{k} g_{j}(\Sigma_{j}) \, . 
\end{align*}
We define $g_{\epsilon}(\Sigma_{\epsilon})$ and $g_{j}(\Sigma_j)$ as follows:
\begin{align*}
g_{\epsilon}(\Sigma_{\epsilon}) := \text{InvW}(\Sigma_{\epsilon} | \nu_{g, \epsilon}, \Sigma_{g, \epsilon}) \, ,
\end{align*}
with 
\begin{align*}
\Sigma_{g,\epsilon} :=  (\nu_{g, \epsilon} + p + 1) \cdot \hat{\Sigma}_{\epsilon}  \, ,
\end{align*}
where $\hat{\Sigma}_{\epsilon}$ is the mode of the posterior probability $p(\Sigma_{\epsilon} | \mathscr{X}, \boldsymbol{\eta}, \mathcal{C})$ (as calculated in the previous section). Note that this choice ensures that the mode of $g_{\epsilon}$ is the same as the mode of 
$p(\Sigma_{\epsilon} | \mathbf{x}_1, ..., \mathbf{x}_n, \boldsymbol{\eta}, \mathcal{C})$.
Analogously, we set
\begin{align*}
g_{j}(\Sigma_{j}) := \text{InvW}(\Sigma_{j} | \nu_{g, j}, \Sigma_{g, j}) \, ,
\end{align*}
with 
\begin{align*}
\Sigma_{g,j} :=  (\nu_{g, j} + p_j + 1) \cdot \hat{\Sigma}_{j}  \, ,
\end{align*}
where $\hat{\Sigma}_{j}$ is the mode of the posterior probability $p(\Sigma_{j} | \mathscr{X}, \boldsymbol{\eta}, \mathcal{C})$.
The remaining parameters $\nu_{g, \epsilon} \in \mathbb{R}$ and $\nu_{g,j} \in \mathbb{R}$ are optimized by minimizing the 
KL-divergence between the the factorized distribution $g$ and the posterior distribution $p(\Sigma_{\epsilon}, \Sigma_{1}, \ldots \Sigma_{k} | \mathbf{x}_1, ..., \mathbf{x}_n, \boldsymbol{\eta}, \mathcal{C})$.  
The details of the following derivations are given in Appendix \ref{app:derivationVA}.
For simplicity let us denote $g_J := \prod_{j=1}^{k} g_{j}$, then we have
\begin{align*}
KL(g || p) &= - \int g_{\epsilon}(\Sigma_{\epsilon}) \cdot \prod_{j=1}^{k} g_{j}(\Sigma_{j}) \\
&\quad \log \frac{p(\Sigma_{\epsilon}, \Sigma_{1}, \ldots \Sigma_{k}, \mathbf{x}_1, ..., \mathbf{x}_n | \boldsymbol{\eta}, \mathcal{C})}{ g_{\epsilon}(\Sigma_{\epsilon}) \cdot \prod_{j=1}^{k} g_{j}(\Sigma_{j})} d \Sigma_{\epsilon} d\Sigma \\
&\quad + c \\
&= - \frac{1}{2} n \E_{g_J, g_{\epsilon}}[\log |\Sigma^{-1} + \beta \Sigma_{\epsilon}^{-1}|] \\
&\quad + \frac{1}{2} (\nu_{\epsilon} + p + 1) \E_{g_{\epsilon}}[ \log |\Sigma_{\epsilon}| ] \\
&\quad+ \frac{1}{2} trace ((\Sigma_{\epsilon,0} + \beta nS) \E_{g_{\epsilon}}[ \Sigma_{\epsilon}^{-1} ] ) \\
&\quad- \text{Entropy}[g_{\epsilon}] \\
&\quad + \frac{1}{2} \sum_{j = 1}^{k}  (\nu_{j} + p_j + 1) \E_{g_j}[ \log |\Sigma_{j}|] \\
&\quad+ \frac{1}{2} \sum_{j = 1}^{k}  trace ((\Sigma_{j,0} + nS_j) \E_{g_j}[\Sigma_{j}^{-1}]) \\
&\quad- \sum_{j = 1}^{k} \text{Entropy}[g_{j}] + c \, ,
\end{align*}
where $c$ is a constant with respect to $g_{\epsilon}$ and $g_j$. 
However, the term $E_{g_J, g_{\epsilon}}[\log |\Sigma^{-1} + \beta \Sigma_{\epsilon}^{-1}|]$ cannot be solved analytically, therefore we need to resort to some sort of approximation. 

We assume that $E_{g_J, g_{\epsilon}}[\log |\Sigma^{-1} + \beta \Sigma_{\epsilon}^{-1}|]$ \\
$ \approx E_{g_J, g_{\epsilon}}[\log |\Sigma^{-1}|]$.
This way, we get 
\begin{align*}
KL(g || p) &\approx  KL(g_{\epsilon} \, || \, \text{InvW}(\nu_{\epsilon}, \Sigma_{\epsilon,0} + \beta nS)) \\
&\quad + \sum_{j = 1}^{k}  KL(g_j \, || \, \text{InvW} (\nu_{j} + n,  \Sigma_{j,0} + nS_j)) \\
&\quad+ c' \, ,
\end{align*}
where we used that 
\begin{align*}
\E_{g_J, g_{\epsilon}}[\log |\Sigma^{-1}|]
= - \sum_{j = 1}^{k}  \E_{g_j}[\log |\Sigma_j|] \, ,
\end{align*}
 and $c'$ is a constant with respect to $g_{\epsilon}$ and $g_j$. 

From the above expression, we see that we can optimize the parameters of $g_{\epsilon}$ and $g_j$ independently from each other.
The optimal parameter $\hat{\nu}_{g, \epsilon}$ for $g_{\epsilon}$ is
\begin{align*}
\hat{\nu}_{g, \epsilon} &= \argmin_{\nu_{g, \epsilon}} KL(g_{\epsilon} \, || \, \text{InvW}(\nu_{\epsilon}, \Sigma_{\epsilon,0} + \beta nS)) \\
&= \argmin_{\nu_{g, \epsilon}} \frac{\nu_{g, \epsilon}}{\nu_{g, \epsilon} + p + 1}  trace ((\Sigma_{\epsilon,0} + \beta nS) \hat{\Sigma}_{\epsilon}^{-1}) \\
&\quad- 2 \log \Gamma_p(\frac{ \nu_{g, \epsilon}}{2}) - \nu_{g, \epsilon} p + p \nu_{\epsilon} \log (\nu_{g, \epsilon} + p + 1) \\
&\quad + (\nu_{g, \epsilon} - \nu_{\epsilon}) \sum_{i=1}^{p} \psi \Big(\frac{\nu_{g, \epsilon} - p + i}{2} \Big) \, .
\end{align*}
And analogously, we have 
\begin{align*}
\hat{\nu}_{g, j} &= \argmin_{\nu_{g, j}} \,  \frac{\nu_{g, j}}{\nu_{g, j} + p_j + 1}  trace ((\Sigma_{j,0} + nS_j) \hat{\Sigma}_{j}^{-1}) \\
&\quad - 2 \log \Gamma_{p_j}(\frac{ \nu_{g, j}}{2}) - \nu_{g, j} p_j \\
&\quad + p_j (\nu_{j} + n) \log (\nu_{g, j} + p_j + 1) \\
&\quad + (\nu_{g, j} - \nu_{j} - n) \sum_{i=1}^{p_j} \psi \Big(\frac{\nu_{g, j} - p_j + i}{2} \Big) \, . 
\end{align*}
Each is a one dimensional non-convex optimization problem that we solve with Brent's method \citep{brent1971algorithms}.

\subsection{MCMC Estimation of Marginal Likelihood} \label{sec:MCMCEstimate}

As an alternative to the variational approximation, we investigate an MCMC estimation based on Chib's method \citep{chib1995marginal,chib2001marginal}.

To simplify the description, we introduction the following notations
\begin{align*}
& \boldsymbol{\theta}_1 := \Sigma_{\epsilon} \, , \\
& \boldsymbol{\theta}_{2}, \ldots,  \boldsymbol{\theta}_{k+1} := \Sigma_1, \ldots,  \Sigma_k \,. 
\end{align*}
Furthermore, we define
$\boldsymbol{\theta}_{< i} := \{\boldsymbol{\theta}_{1}, \ldots, \boldsymbol{\theta}_{i - 1} \}$ and $\boldsymbol{\theta}_{> i} := \{\boldsymbol{\theta}_{i+1}, \ldots, \boldsymbol{\theta}_{k+1} \}$.  
For simplicity, we also suppress in the notation the explicit conditioning on the hyper-parameters $\boldsymbol{\eta}$ and the clustering $\mathcal{C}$, which are both fixed.

Following the strategy of \cite{chib1995marginal}, the marginal likelihood can be expressed as
\begin{equation}
\begin{aligned} \label{eq:marginalAroundModeChib}
p(\mathscr{X}) &= \frac{p(\hat{\boldsymbol{\theta}}_1, \ldots, \hat{\boldsymbol{\theta}}_{k+1} , \mathscr{X})}{p(\hat{\boldsymbol{\theta}}_1, \ldots, \hat{\boldsymbol{\theta}}_{k+1} | \mathscr{X}) } \\
&= \frac{p(\hat{\boldsymbol{\theta}}_1, \ldots, \hat{\boldsymbol{\theta}}_{k+1} , \mathscr{X})}{ \prod_{i = 1}^{k+1} p(\hat{\boldsymbol{\theta}}_i | \mathscr{X}, \hat{\boldsymbol{\theta}}_{1}  \ldots, \hat{\boldsymbol{\theta}}_{i-1}) } 
\end{aligned}
\end{equation}

In order to approximate $p(\mathscr{X})$ with Equation \eqref{eq:marginalAroundModeChib}, we need to estimate $p(\hat{\boldsymbol{\theta}}_i | \mathscr{X}, \hat{\boldsymbol{\theta}}_1, \ldots \hat{\boldsymbol{\theta}}_{i-1})$.
First, note that we can express the value of the conditional posterior distribution at $\hat{\boldsymbol{\theta}}_i$, as follows (see \cite{chib2001marginal}, Section 2.3): \
\begin{equation}
\begin{aligned} \label{eq:conditionalProbValueExpected}
&p(\hat{\boldsymbol{\theta}}_i | \mathscr{X}, \hat{\boldsymbol{\theta}}_1, \ldots \hat{\boldsymbol{\theta}}_{i-1}) \\
&= \frac{\E_{\boldsymbol{\theta}_{\geq i}  \sim p(\boldsymbol{\theta}_{\geq i} | \mathscr{X}, \hat{\boldsymbol{\theta}}_{< i})} [ \alpha(\boldsymbol{\theta}_i,  \hat{\boldsymbol{\theta}}_i |  \hat{\boldsymbol{\theta}}_{< i}, \boldsymbol{\theta}_{> i})  q_i(\hat{\boldsymbol{\theta}}_i)] }
{\E_{\boldsymbol{\theta}_{\geq i}  \sim p(\boldsymbol{\theta}_{> i} | \mathscr{X}, \hat{\boldsymbol{\theta}}_{\leq i}) q(\boldsymbol{\theta}_i)} [ \alpha(\hat{\boldsymbol{\theta}}_i, \boldsymbol{\theta}_i |  \hat{\boldsymbol{\theta}}_{< i}, \boldsymbol{\theta}_{> i})] } \, ,
\end{aligned}
\end{equation}
where $q_i(\boldsymbol{\theta}_i)$ is a proposal distribution for $\boldsymbol{\theta}_i$, and the acceptance probability of moving from state $\boldsymbol{\theta}_i$ to state $\boldsymbol{\theta}_i'$, holding the other states fixed is defined as
\begin{align} \label{eq:acceptanceProb}
\alpha(\boldsymbol{\theta}_i,  \boldsymbol{\theta}_i' |  \boldsymbol{\theta}_{< i}, \boldsymbol{\theta}_{> i}) := \min \{ 1, \frac{p(\mathscr{X}, \boldsymbol{\theta}_{< i}, \boldsymbol{\theta}_{> i}, \boldsymbol{\theta}_i') \cdot q_i(\boldsymbol{\theta}_i) }
{p(\mathscr{X}, \boldsymbol{\theta}_{< i}, \boldsymbol{\theta}_{> i}, \boldsymbol{\theta}_i) \cdot q_i(\boldsymbol{\theta}_i') } \} \, .
\end{align}

Next, using Equation \eqref{eq:conditionalProbValueExpected}, we can estimate \\
$p(\hat{\boldsymbol{\theta}}_i | \mathscr{X}, \hat{\boldsymbol{\theta}}_1, \ldots \hat{\boldsymbol{\theta}}_{i-1})$ with a Monte Carlo approximation with $M$ samples:

\begin{equation}
\begin{aligned}  \label{eq:conditionalProbValueMCMCestimate}
& p(\hat{\boldsymbol{\theta}}_i | \mathscr{X}, \hat{\boldsymbol{\theta}}_1, \ldots \hat{\boldsymbol{\theta}}_{i-1}) \\
& \approx 
\frac{ \frac{1}{M} \sum_{m=1}^M \alpha(\boldsymbol{\theta}_i^{i,m},  \hat{\boldsymbol{\theta}}_i |  \hat{\boldsymbol{\theta}}_{< i}, \boldsymbol{\theta}_{> i}^{i,m})  q_i(\hat{\boldsymbol{\theta}}_i) }
{\frac{1}{M}  \sum_{m=1}^M  \alpha(\hat{\boldsymbol{\theta}}_i, \boldsymbol{\theta}_i^{q,m} |  \hat{\boldsymbol{\theta}}_{< i}, \boldsymbol{\theta}_{> i}^{i+1, m} ) }
\end{aligned} 
\end{equation}
where $\boldsymbol{\theta}_{i}^{a, m} \sim p(\boldsymbol{\theta}_{i} | \mathscr{X}, \hat{\boldsymbol{\theta}}_{< a})$, $\boldsymbol{\theta}_{> i}^{a, m} \sim p(\boldsymbol{\theta}_{> i} | \mathscr{X}, \hat{\boldsymbol{\theta}}_{< a})$, and $\boldsymbol{\theta}_{i}^{q,m} \sim q(\boldsymbol{\theta}_i)$.

Finally, in order to sample from $p(\boldsymbol{\theta}_{\geq i} | \mathscr{X}, \hat{\boldsymbol{\theta}}_{< i})$, we propose to use the Metropolis-Hastings within Gibbs sampler as shown in Algorithm \ref{alg:propMHGibs}.
$MH_j (\boldsymbol{\theta}_j^{t},  \boldsymbol{\psi})$ denotes the Metropolis-Hastings algorithm with current state $\boldsymbol{\theta}_j^{t}$, and acceptance probability $\alpha(\boldsymbol{\theta}_j,  \boldsymbol{\theta}_j' |  \boldsymbol{\psi})$, Equation \eqref{eq:acceptanceProb}, and $\boldsymbol{\theta}_{\geq i}^{0}$ is a sample after the burn-in.
For the proposal distribution $q_i(\boldsymbol{\theta}_i)$, we use 
\begin{align} \label{eq:proposalDist}
q_i := 
\left\{
\begin{array}{ll}
\text{InvW}(\nu, \hat{\Sigma}_{\epsilon} \cdot (\nu + p + 1))  \\
\text{with $\nu = \beta \kappa \cdot n + \nu_{\epsilon}$}  & \text{if } i = 1,\\
\text{InvW}(\nu, \hat{\Sigma}_{i-1} \cdot (\nu + p_{i-1} + 1)) \\
 \text{with $\nu = (1 - \beta) \kappa \cdot n + \nu_{i-1}$}  & \text{else. }
\end{array} \right.
\end{align}
Here $\kappa >0$ is a hyper-parameter of the MCMC algorithm that is chosen to control the acceptance probability.
Note that if we choose $\kappa = 1$ and $\beta$ is 0, then the proposal distribution $q_i(\boldsymbol{\theta}_i)$ equals the posterior distribution $p(\boldsymbol{\theta}_i | \mathscr{X}, \hat{\boldsymbol{\theta}}_1, \ldots \hat{\boldsymbol{\theta}}_{i-1})$.
However, in practice, we found that the acceptance probabilities can be too small, leading to unstable estimates and division by 0 in Equation \eqref{eq:conditionalProbValueMCMCestimate}. Therefore, for our experiments we chose $\kappa = 10$.

\begin{algorithm}[h]
\caption{Metropolis-Hastings within Gibbs sampler for sampling from $p(\boldsymbol{\theta}_{\geq i} | \mathscr{X}, \hat{\boldsymbol{\theta}}_{< i})$. \label{alg:propMHGibs}}
\begin{algorithmic}
\FOR {$t$ from 1 to $M$}
\FOR  {$j$ from $i$ to $k+1$}
	\STATE $\boldsymbol{\psi} := \{\hat{\boldsymbol{\theta}}_{<i}, \boldsymbol{\theta}_{i}^{t}, \ldots, \boldsymbol{\theta}_{j-1}^{t}, \boldsymbol{\theta}_{> j}^{t-1} \}$
	\STATE $\boldsymbol{\theta}_j^{t} := MH_j (\boldsymbol{\theta}_j^{t-1},  \boldsymbol{\psi})$
\ENDFOR
\ENDFOR
\end{algorithmic}
\end{algorithm}

\section{Restricting the hypotheses space} \label{sec:restrictingHypothesesSpace}

The number of possible clusterings follow the Bell numbers, and therefore it is infeasible to enumerate all possible clusterings, even if the number of variables $p$ is small.
It is therefore crucial to restrict the hypotheses space to a subset of all clusterings that are likely to contain the true clustering.
We denote this subset as $\mathscr{C}^*$.

We suggest to use spectral clustering on different estimates of the precision matrix to acquire the set of clusterings $\mathscr{C}^*$.
A motivation for this heuristic is given in Appendix \ref{app:spectralClustering}.

First, for an appropriate $\lambda$, we estimate the precision matrix using
\begin{align} \label{eq:finalOptProblem}
& X^{*} := \argmin_{X \succeq 0} - \log |X| + trace (X S ) + \lambda \sum_{i \neq j} |X_{ij}|^q \, .
\end{align}
In our experiments, we take $q = 1$, which is equivalent to the Graphical Lasso \citep{friedman2008sparse} with an l1-penalty on all entries of $X$ except the diagonal.
In the next step, we then construct the Laplacian $L$ as defined in the following.
\begin{equation}
\begin{aligned} \label{eq:laplacianDef}
& L_{ii} = \sum_{k \neq i} |X_{ik}^*|^q \, , \\
& L_{ij} =  - |X_{ij}^*|^q  \; \;  \text{for } i \neq j \, . 
\end{aligned}
\end{equation}
Finally, we use k-means clustering on the eigenvectors of the Laplacian L.
The details of acquiring the set of clusterings $\mathscr{C}^*$ using the spectral clustering method are summarized below:
\begin{algorithm}
\caption{Spectral Clustering for variable clustering with the Gaussian graphical model. \label{alg:spectralVarClustering}}
\begin{algorithmic}
\STATE $J$ := set of regularization parameter values.
\STATE $K_{max}$ := maximum number of considered clusters.
\STATE $\mathscr{C}^* := \{ \}$
\FOR  {$\lambda \in J$}
\STATE $X^* := $ solve optimization problem from Equation \eqref{eq:finalOptProblem}.
\STATE $(\mathbf{e}_1,\ldots, \mathbf{e}_{K_{max}}) := $ determine the eigenvectors corresponding to the $K_{max}$ lowest eigenvalues of the Laplacian $L$ as defined in Equations \eqref{eq:laplacianDef}.
\FOR  {$k \in \{2,\ldots ,K_{max}\}$}
\STATE $\mathcal{C}_{\lambda, k} := $ cluster all variables into $k$ partitions using k-means with $(\mathbf{e}_1,\ldots, \mathbf{e}_k)$.
\STATE $\mathscr{C}^* := \mathscr{C}^* \cup \mathcal{C}_{\lambda, k}$
\ENDFOR
\ENDFOR
\RETURN restricted hypotheses space  $\mathscr{C}^*$
\end{algorithmic}
\end{algorithm}

In Section \ref{sec:evalOfRestrictedHypothesesSpace} we confirm experimentally that, even in the presence of noise, $\mathscr{C}^*$ often contains the true clustering, or clusterings that are close to the true clustering.

\subsection{Posterior distribution over number of clusters} \label{sec:posteriorOverNumberOfClusters}

In principle, the posterior distribution for the number of clusters can be calculated using
\begin{align*}
p(k | \mathscr{X}) \propto \sum_{\mathcal{C} \in \mathscr{C}_{k}} p(\mathscr{X} | \mathcal{C}) \, ,
\end{align*}
where $\mathscr{C}_{k}$ denotes the set of all clusterings with number of clusters being equal to $k$. 
Since this is computationally infeasible, we use the following approximation
\begin{align*}
P(k | X) \propto \sum_{\mathcal{C} \in \mathscr{C}_{k}}  p(X | \mathcal{C}) \approx \sum_{\mathcal{C} \in \mathscr{C}^*_{k}}  p(X | \mathcal{C}) \, ,
\end{align*}
where $\mathscr{C}^*_{k}$ is the set of all clusterings with $k$ clusters that are in the restricted hypotheses space $\mathscr{C}^*$.

\section{Simulation Study} \label{sec:simulatedData}

In this section, we evaluate the proposed method on simulated data for which the ground truth is available. In sub-section \ref{sec:evalOfRestrictedHypothesesSpace}, we evaluate the quality of the restricted hypotheses space $\mathscr{C}^*$,
followed by sub-section \ref{sec:evalClusteringSelectionSyntheticData}, where we evaluated the proposed method's ability to select the best clustering in $\mathscr{C}^*$.

For the number of clusters we consider the range from $2$ to $15$. 
For the set of regularization parameters of the spectral clustering method we use $J := \{$0.0001, 0.0005, 0.001, 0.002, 0.003, 0.004, 0.005, 0.006, 0.007, 0.008, 0.009, 0.01$\}$ (see Algorithm \ref{alg:spectralVarClustering}).

In all experiments the number of variables is $p = 40$, and the ground truth is 4 clusters with 10 variables each.

For generating positive-definite covariance matrices, we consider the following two distributions: $\text{InvW}(p + 1, I_{p})$, and $\text{Uniform}_p$, with dimension $p$.
We denote by $U \sim \text{Uniform}_p$ the positive-definite matrix generated in the following way
\begin{align*}
& U = A + (0.001 - \lambda_{min} (A)) I_p \, ,
\end{align*}
where $\lambda_{min} (A)$ is the smallest eigenvalue of $A$, and $A$ is drawn as follows
\begin{align*}
& A_{i,j} = A_{j,i}  \sim \text{Uniform}(-1, 1)  \,  , i \neq j \\
& A_{i,i} = 0 \, .
\end{align*}

For generating $\Sigma$, we either sample each block $j$ from $\text{InvW}(p_j + 1, I_{p_j})$ or from $\text{Uniform}_{p_j}$. 

For generating the noise matrix $\Sigma_{\epsilon}$, we sample either from $\text{InvW}(p + 1, I_{p})$ or from $\text{Uniform}_{p}$.
The final data is then sampled as follows
\begin{align*}
x \sim N(0, (\Sigma^{-1} + \eta \Sigma_{\epsilon}^{-1})^{-1}) \, ,
\end{align*}
where $\eta$ defines the noise level.

For evaluation we use the adjusted normalized mutual information (ANMI), where $0.0$ means that any correspondence with the true labels is at chance level, and $1.0$ means that a perfect one-to-one correspondence exists \citep{vinh2010information}.
We repeated all experiments 5 times and report the average ANMI score.

\subsection{Evaluation of the restricted hypotheses space} \label{sec:evalOfRestrictedHypothesesSpace}

First, independent of any model selection criteria, we check here the quality of the clusterings that are found with the spectral clustering algorithm from Section \ref{sec:restrictingHypothesesSpace}.
We also compare to single and average linkage clustering as used in \citep{tan2015cluster}.  

The set of all clusterings that are found is denoted by $\mathscr{C}^*$ (the restricted hypotheses space).

In order to evaluate the quality of the restricted hypotheses space $\mathscr{C}^*$, we report the oracle performance calculated by 
$\max_{\mathcal{C} \in \mathscr{C}^*} \text{ANMI}(\mathcal{C}, \mathcal{C}_T)$,
where $\mathcal{C}_T$ denotes the true clustering, and $\text{ANMI}(\mathcal{C}, \mathcal{C}_T)$ denotes the ANMI score when comparing clustering $\mathcal{C}$ with the true clustering.
In particular, a score of 1.0 means that the true clustering is contained in $\mathscr{C}^*$.

The results of all experiments with noise level $\eta \in \{0.0, 0.01, 0.1\}$ are shown in Tables \ref{tab:hypothesesSpaceEvaluationBalanced}, for balanced clusters, and Table \ref{tab:hypothesesSpaceEvaluationUnbalanced}, for unbalanced clusters.

From these results we see that the restricted hypotheses space of spectral clustering is around 100, considerably smaller than the number of all possible clusterings. 
More importantly, we also see that that $\mathscr{C}^*$ acquired by spectral clustering either contains the true clustering or a clustering that is close to the truth.
In contrast, the hypotheses space restricted by single and average linkage is smaller, but more often misses the true clustering. 

\begin{table*}[h]
\footnotesize
  \caption{Evaluation of restricted hypotheses space for $p = 40$, 
  $n \in \{20, 40, 400, 4000, 40000, 4000000\}$. Ground truth contains 4 balanced clusters. 
  Shows the oracle performance measured by ANMI for spectral clustering, average linkage and single linkage. Note that that an ANMI score of 1.0 means that the true clustering is contained in the hypotheses space found by the clustering method.
  The size of the hypotheses space restricted by each clustering method is denoted by $|\mathscr{C}^*|$.
  Average results over 5 runs with standard deviation in brackets.}
  \label{tab:hypothesesSpaceEvaluationBalanced}
  \begin{tabular}{lcllllll}
  \toprule
    & \multicolumn{7}{c}{$\Sigma_j \sim \text{InvW}(p_j + 1, I_{p_j})$, no noise} \\
  \cmidrule{2-8}
   & & 20 & 40 & 400 & 4000 & 40000 & 4000000 \\
 \cmidrule{2-8}
\multirow{2}{*}{spectral}  & ANMI  &    0.77 (0.14) & 0.95 (0.06) & 1.0 (0.0) & 1.0 (0.0) & 1.0 (0.0) & 1.0 (0.0) \\
 & $|\mathscr{C}^*|$ & 140.8 (5.78) & 139.0 (8.65) & 112.8 (5.64) & 99.8 (2.23) & 101.4 (7.94) & 98.4 (3.61) \\
 \cmidrule{2-8}
\multirow{2}{*}{average}  & ANMI  &     0.38 (0.09) & 0.38 (0.06) & 0.45 (0.05) & 0.45 (0.03) & 0.45 (0.07) & 0.45 (0.03) \\
 & $|\mathscr{C}^*|$ & 14.0 (0.0) & 14.0 (0.0) & 14.0 (0.0) & 14.0 (0.0) & 14.0 (0.0) & 14.0 (0.0) \\
\cmidrule{2-8}
\multirow{2}{*}{single}  & ANMI  &      0.32 (0.08) & 0.34 (0.09) & 0.39 (0.08) & 0.39 (0.08) & 0.42 (0.14) & 0.41 (0.08) \\
 & $|\mathscr{C}^*|$ & 14.0 (0.0) & 14.0 (0.0) & 14.0 (0.0) & 14.0 (0.0) & 14.0 (0.0) & 14.0 (0.0) \\
\midrule
& \multicolumn{6}{c}{$\Sigma_j \sim \text{InvW}(p_j + 1, I_{p_j}), \Sigma_{\epsilon} \sim \text{InvW}(p + 1, I_{p}), \eta = 0.01$} \\
 \cmidrule{2-8}
\multirow{2}{*}{spectral}  & ANMI  &    0.49 (0.03) & 0.9 (0.03) & 1.0 (0.0) & 1.0 (0.0) & 1.0 (0.0) & 1.0 (0.0) \\
 & $|\mathscr{C}^*|$  & 143.2 (7.25) & 144.4 (3.32) & 108.6 (9.89) & 105.4 (9.79) & 103.6 (5.0) & 97.0 (6.57) \\
\cmidrule{2-8}
\multirow{2}{*}{average}  & ANMI  &     0.26 (0.05) & 0.34 (0.04) & 0.46 (0.07) & 0.51 (0.08) & 0.42 (0.09) & 0.45 (0.06) \\
 & $|\mathscr{C}^*|$ & 14.0 (0.0) & 14.0 (0.0) & 14.0 (0.0) & 14.0 (0.0) & 14.0 (0.0) & 14.0 (0.0) \\
 \cmidrule{2-8}
\multirow{2}{*}{single}  & ANMI  &      0.16 (0.08) & 0.25 (0.08) & 0.37 (0.03) & 0.4 (0.06) & 0.3 (0.12) & 0.32 (0.09) \\
 & $|\mathscr{C}^*|$ & 14.0 (0.0) & 14.0 (0.0) & 14.0 (0.0) & 14.0 (0.0) & 14.0 (0.0) & 14.0 (0.0) \\
\midrule
 & \multicolumn{6}{c}{$\Sigma_j \sim \text{InvW}(p_j + 1, I_{p_j}), \Sigma_{\epsilon} \sim \text{InvW}(p + 1, I_{p}), \eta = 0.1$} \\
 \cmidrule{2-8}
\multirow{2}{*}{spectral}  & ANMI  &    0.34 (0.1) & 0.87 (0.09) & 1.0 (0.0) & 1.0 (0.0) & 1.0 (0.0) & 1.0 (0.0) \\
 & $|\mathscr{C}^*|$ & 121.4 (7.34) & 106.4 (18.51) & 35.4 (5.12) & 33.2 (11.48) & 37.4 (5.54) & 31.0 (8.65) \\
\cmidrule{2-8}
\multirow{2}{*}{average}  & ANMI  &     0.1 (0.05) & 0.15 (0.03) & 0.34 (0.08) & 0.37 (0.1) & 0.26 (0.11) & 0.28 (0.09) \\
 & $|\mathscr{C}^*|$ & 14.0 (0.0) & 14.0 (0.0) & 14.0 (0.0) & 14.0 (0.0) & 14.0 (0.0) & 14.0 (0.0) \\
 \cmidrule{2-8}
\multirow{2}{*}{single}  & ANMI  &      0.04 (0.03) & 0.08 (0.04) & 0.19 (0.11) & 0.21 (0.06) & 0.11 (0.03) & 0.13 (0.02) \\
 & $|\mathscr{C}^*|$ & 14.0 (0.0) & 14.0 (0.0) & 14.0 (0.0) & 14.0 (0.0) & 14.0 (0.0) & 14.0 (0.0) \\
\midrule
    & \multicolumn{7}{c}{$\Sigma_j \sim \text{Uniform}_{p_j}$, no noise} \\
     \cmidrule{2-8}
\multirow{2}{*}{spectral}  & ANMI  &    0.34 (0.1) & 0.87 (0.09) & 1.0 (0.0) & 1.0 (0.0) & 1.0 (0.0) & 1.0 (0.0) \\
 & $|\mathscr{C}^*|$ & 121.4 (7.34) & 106.4 (18.51) & 35.4 (5.12) & 33.2 (11.48) & 37.4 (5.54) & 31.0 (8.65) \\
\cmidrule{2-8}
\multirow{2}{*}{average}  & ANMI  &     0.1 (0.06) & 0.26 (0.07) & 0.92 (0.11) & 1.0 (0.0) & 1.0 (0.0) & 0.99 (0.03) \\
 & $|\mathscr{C}^*|$ & 14.0 (0.0) & 14.0 (0.0) & 14.0 (0.0) & 14.0 (0.0) & 14.0 (0.0) & 14.0 (0.0) \\
\cmidrule{2-8}
\multirow{2}{*}{single}  & ANMI  &      0.04 (0.02) & 0.13 (0.08) & 0.82 (0.25) & 1.0 (0.0) & 1.0 (0.0) & 0.99 (0.03) \\
 & $|\mathscr{C}^*|$ & 14.0 (0.0) & 14.0 (0.0) & 14.0 (0.0) & 14.0 (0.0) & 14.0 (0.0) & 14.0 (0.0) \\
\midrule
& \multicolumn{6}{c}{$\Sigma_j \sim \text{Uniform}_{p_j}, \Sigma_{\epsilon} \sim \text{Uniform}_{p}, \eta = 0.01$} \\
 \cmidrule{2-8}
\multirow{2}{*}{spectral}  & ANMI  &    0.28 (0.06) & 0.81 (0.1) & 0.94 (0.06) & 0.99 (0.03) & 0.99 (0.03) & 0.97 (0.03) \\
 & $|\mathscr{C}^*|$ & 127.2 (3.6) & 106.0 (5.29) & 48.2 (9.77) & 50.2 (5.95) & 51.0 (8.94) & 48.0 (5.69) \\
\cmidrule{2-8}
\multirow{2}{*}{average}  & ANMI  &     0.14 (0.05) & 0.22 (0.04) & 0.81 (0.16) & 0.89 (0.1) & 0.87 (0.12) & 0.94 (0.12) \\
 & $|\mathscr{C}^*|$ & 14.0 (0.0) & 14.0 (0.0) & 14.0 (0.0) & 14.0 (0.0) & 14.0 (0.0) & 14.0 (0.0) \\
\cmidrule{2-8}
\multirow{2}{*}{single}  & ANMI  &      0.04 (0.02) & 0.1 (0.04) & 0.78 (0.13) & 0.71 (0.23) & 0.78 (0.11) & 0.79 (0.17) \\
 & $|\mathscr{C}^*|$ & 14.0 (0.0) & 14.0 (0.0) & 14.0 (0.0) & 14.0 (0.0) & 14.0 (0.0) & 14.0 (0.0) \\
\midrule
& \multicolumn{6}{c}{$\Sigma_j \sim \text{Uniform}_{p_j}, \Sigma_{\epsilon} \sim \text{Uniform}_{p}, \eta = 0.1$} \\
 \cmidrule{2-8}
\multirow{2}{*}{spectral}  & ANMI  &    0.3 (0.03) & 0.72 (0.08) & 0.88 (0.07) & 0.9 (0.07) & 0.87 (0.11) & 0.88 (0.04) \\
 & $|\mathscr{C}^*|$ & 126.2 (2.23) & 120.4 (9.35) & 74.4 (19.41) & 87.2 (7.93) & 79.2 (13.61) & 77.0 (14.25) \\
 \cmidrule{2-8}
\multirow{2}{*}{average}  & ANMI  &     0.08 (0.04) & 0.26 (0.11) & 0.83 (0.15) & 0.88 (0.12) & 0.87 (0.11) & 0.94 (0.12) \\
 & $|\mathscr{C}^*|$ & 14.0 (0.0) & 14.0 (0.0) & 14.0 (0.0) & 14.0 (0.0) & 14.0 (0.0) & 14.0 (0.0) \\
\cmidrule{2-8}
\multirow{2}{*}{single}  & ANMI  &      0.05 (0.03) & 0.13 (0.07) & 0.7 (0.14) & 0.69 (0.15) & 0.76 (0.12) & 0.76 (0.14) \\
 & $|\mathscr{C}^*|$ & 14.0 (0.0) & 14.0 (0.0) & 14.0 (0.0) & 14.0 (0.0) & 14.0 (0.0) & 14.0 (0.0) \\
\bottomrule
    \end{tabular}
 \end{table*}

 \begin{table*}[h]
\footnotesize
  \caption{Same setting as in Table \ref{tab:hypothesesSpaceEvaluationBalanced} but with unbalanced clusters. Ground truth is 4 clusters with sizes 20, 10, 5, 5.}
  \label{tab:hypothesesSpaceEvaluationUnbalanced}
  \begin{tabular}{lcllllll}
  \toprule
    & \multicolumn{7}{c}{$\Sigma_j \sim \text{InvW}(p_j + 1, I_{p_j})$, no noise} \\
  \cmidrule{2-8}
   & & 20 & 40 & 400 & 4000 & 40000 & 4000000 \\
\cmidrule{2-8}
\multirow{2}{*}{spectral}  & ANMI  &    0.52 (0.13) & 0.85 (0.11) & 1.0 (0.0) & 1.0 (0.0) & 1.0 (0.0) & 1.0 (0.0) \\
 & $|\mathscr{C}^*|$ & 141.2 (6.62) & 133.2 (8.03) & 80.8 (8.21) & 73.4 (8.89) & 62.0 (7.38) & 62.6 (7.23) \\
 \cmidrule{2-8}
\multirow{2}{*}{average}  & ANMI  &     0.34 (0.06) & 0.39 (0.05) & 0.37 (0.04) & 0.38 (0.07) & 0.38 (0.06) & 0.44 (0.09) \\
 & $|\mathscr{C}^*|$ & 14.0 (0.0) & 14.0 (0.0) & 14.0 (0.0) & 14.0 (0.0) & 14.0 (0.0) & 14.0 (0.0) \\
\cmidrule{2-8}
\multirow{2}{*}{single}  & ANMI  &      0.33 (0.05) & 0.35 (0.03) & 0.32 (0.04) & 0.32 (0.14) & 0.27 (0.13) & 0.39 (0.12) \\
 & $|\mathscr{C}^*|$ & 14.0 (0.0) & 14.0 (0.0) & 14.0 (0.0) & 14.0 (0.0) & 14.0 (0.0) & 14.0 (0.0) \\
\midrule
& \multicolumn{6}{c}{$\Sigma_j \sim \text{InvW}(p_j + 1, I_{p_j}), \Sigma_{\epsilon} \sim \text{InvW}(p + 1, I_{p}), \eta = 0.01$} \\
\cmidrule{2-8}
\multirow{2}{*}{spectral}  & ANMI  &    0.55 (0.13) & 0.81 (0.07) & 1.0 (0.0) & 1.0 (0.0) & 1.0 (0.0) & 1.0 (0.0) \\
 & $|\mathscr{C}^*|$ & 148.8 (4.62) & 136.0 (6.81) & 80.4 (9.77) & 68.8 (10.3) & 67.0 (5.93) & 63.0 (14.3) \\
 \cmidrule{2-8}
\multirow{2}{*}{average}  & ANMI  &     0.34 (0.06) & 0.37 (0.08) & 0.53 (0.12) & 0.5 (0.1) & 0.46 (0.1) & 0.52 (0.1) \\
 & $|\mathscr{C}^*|$ & 14.0 (0.0) & 14.0 (0.0) & 14.0 (0.0) & 14.0 (0.0) & 14.0 (0.0) & 14.0 (0.0) \\
 \cmidrule{2-8}
\multirow{2}{*}{single}  & ANMI  &      0.29 (0.07) & 0.29 (0.08) & 0.41 (0.17) & 0.4 (0.14) & 0.37 (0.11) & 0.32 (0.12) \\
 & $|\mathscr{C}^*|$ & 14.0 (0.0) & 14.0 (0.0) & 14.0 (0.0) & 14.0 (0.0) & 14.0 (0.0) & 14.0 (0.0) \\
\midrule
 & \multicolumn{6}{c}{$\Sigma_j \sim \text{InvW}(p_j + 1, I_{p_j}), \Sigma_{\epsilon} \sim \text{InvW}(p + 1, I_{p}), \eta = 0.1$} \\
\cmidrule{2-8}
\multirow{2}{*}{spectral}  & ANMI  &    0.26 (0.04) & 0.5 (0.06) & 0.93 (0.07) & 0.93 (0.07) & 0.99 (0.02) & 0.91 (0.08) \\
 & $|\mathscr{C}^*|$ & 144.4 (5.54) & 159.2 (1.83) & 121.0 (10.43) & 120.2 (6.62) & 117.0 (3.41) & 113.2 (11.91) \\
 \cmidrule{2-8}
\multirow{2}{*}{average}  & ANMI  &     0.2 (0.03) & 0.22 (0.06) & 0.37 (0.09) & 0.36 (0.08) & 0.41 (0.13) & 0.44 (0.07) \\
 & $|\mathscr{C}^*|$ & 14.0 (0.0) & 14.0 (0.0) & 14.0 (0.0) & 14.0 (0.0) & 14.0 (0.0) & 14.0 (0.0) \\
\cmidrule{2-8}
\multirow{2}{*}{single}  & ANMI  &      0.2 (0.08) & 0.2 (0.07) & 0.24 (0.04) & 0.29 (0.05) & 0.33 (0.07) & 0.32 (0.05) \\
 & $|\mathscr{C}^*|$ & 14.0 (0.0) & 14.0 (0.0) & 14.0 (0.0) & 14.0 (0.0) & 14.0 (0.0) & 14.0 (0.0) \\
\midrule
    & \multicolumn{7}{c}{$\Sigma_j \sim \text{Uniform}_{p_j}$, no noise} \\
\cmidrule{2-8}
\multirow{2}{*}{spectral}  & ANMI  &    0.36 (0.06) & 0.72 (0.13) & 1.0 (0.0) & 1.0 (0.0) & 1.0 (0.0) & 1.0 (0.0) \\
 & $|\mathscr{C}^*|$ & 124.0 (7.29) & 115.8 (9.89) & 40.8 (12.5) & 39.4 (5.2) & 33.2 (4.79) & 38.6 (5.24) \\
\cmidrule{2-8}
\multirow{2}{*}{average}  & ANMI  &     0.09 (0.04) & 0.05 (0.08) & 0.12 (0.07) & 0.29 (0.07) & 0.37 (0.07) & 0.34 (0.14) \\
 & $|\mathscr{C}^*|$ & 14.0 (0.0) & 14.0 (0.0) & 14.0 (0.0) & 14.0 (0.0) & 14.0 (0.0) & 14.0 (0.0) \\
\cmidrule{2-8}
\multirow{2}{*}{single}  & ANMI  &      0.01 (0.04) & -0.01 (0.0) & -0.01 (0.01) & 0.06 (0.1) & 0.17 (0.19) & 0.13 (0.12) \\
 & $|\mathscr{C}^*|$ & 14.0 (0.0) & 14.0 (0.0) & 14.0 (0.0) & 14.0 (0.0) & 14.0 (0.0) & 14.0 (0.0) \\
\midrule
& \multicolumn{6}{c}{$\Sigma_j \sim \text{Uniform}_{p_j}, \Sigma_{\epsilon} \sim \text{Uniform}_{p}, \eta = 0.01$} \\
\cmidrule{2-8}
\multirow{2}{*}{spectral}  & ANMI  &    0.39 (0.04) & 0.67 (0.11) & 0.85 (0.05) & 0.89 (0.07) & 0.87 (0.07) & 0.89 (0.06) \\
 & $|\mathscr{C}^*|$ & 125.6 (8.06) & 115.0 (12.85) & 42.6 (7.09) & 59.2 (11.55) & 53.2 (9.2) & 54.0 (6.69) \\
\cmidrule{2-8}
\multirow{2}{*}{average}  & ANMI  &     0.04 (0.03) & 0.06 (0.05) & 0.12 (0.06) & 0.21 (0.08) & 0.18 (0.09) & 0.21 (0.13) \\
 & $|\mathscr{C}^*|$ & 14.0 (0.0) & 14.0 (0.0) & 14.0 (0.0) & 14.0 (0.0) & 14.0 (0.0) & 14.0 (0.0) \\
\cmidrule{2-8}
\multirow{2}{*}{single}  & ANMI  &      -0.01 (0.0) & -0.01 (0.0) & -0.01 (0.0) & 0.0 (0.02) & 0.01 (0.05) & 0.02 (0.05) \\
 & $|\mathscr{C}^*|$ & 14.0 (0.0) & 14.0 (0.0) & 14.0 (0.0) & 14.0 (0.0) & 14.0 (0.0) & 14.0 (0.0) \\
\midrule
& \multicolumn{6}{c}{$\Sigma_j \sim \text{Uniform}_{p_j}, \Sigma_{\epsilon} \sim \text{Uniform}_{p}, \eta = 0.1$} \\
\cmidrule{2-8}
\multirow{2}{*}{spectral}  & ANMI  &    0.32 (0.06) & 0.68 (0.13) & 0.8 (0.09) & 0.81 (0.09) & 0.79 (0.07) & 0.78 (0.09) \\
 & $|\mathscr{C}^*|$ & 124.2 (9.33) & 109.6 (12.63) & 66.6 (10.71) & 74.2 (7.14) & 62.8 (5.11) & 65.2 (13.85) \\
 \cmidrule{2-8}
\multirow{2}{*}{average}  & ANMI  &     0.04 (0.03) & 0.06 (0.05) & 0.09 (0.05) & 0.19 (0.05) & 0.13 (0.06) & 0.2 (0.13) \\
 & $|\mathscr{C}^*|$ & 14.0 (0.0) & 14.0 (0.0) & 14.0 (0.0) & 14.0 (0.0) & 14.0 (0.0) & 14.0 (0.0) \\
\cmidrule{2-8}
\multirow{2}{*}{single}  & ANMI  &      -0.01 (0.0) & -0.01 (0.0) & -0.01 (0.0) & -0.01 (0.0) & -0.01 (0.0) & 0.0 (0.02) \\
 & $|\mathscr{C}^*|$ & 14.0 (0.0) & 14.0 (0.0) & 14.0 (0.0) & 14.0 (0.0) & 14.0 (0.0) & 14.0 (0.0) \\
\bottomrule
    \end{tabular}
 \end{table*}

\subsection{Evaluation of clustering selection criteria} \label{sec:evalClusteringSelectionSyntheticData}

Here, we evaluate the performance of our proposed method for selecting the correct clustering in the restricted hypotheses space $\mathscr{C}^*$.
We compare our proposed method (variational) with several baselines and two previously proposed methods \citep{tan2015cluster,palla2012nonparametric}.
Except for the two previously proposed methods, we created $\mathscr{C}^*$ with the spectral clustering algorithm from Section \ref{sec:restrictingHypothesesSpace}.

As a cluster selection criteria, we compare our method to the Extended Bayesian Information Criterion (EBIC) with $\gamma \in \{0, 0.5, 1\}$ \citep{chen2008extended,foygel2010extended}, Akaike Information Criteria \citep{akaike1973information}, and the Calinski-Harabaz Index (CHI) \citep{calinski1974dendrite}. 
Note that EBIC and AIC are calculated based on the basic Gaussian graphical model (i.e. the model in Equation \ref{eq:basicModel}, but ignoring the prior specification).\footnote{As discussed in Section \ref{sec:estimationMarginalLikelihood}, EBIC (and also AIC) cannot be used with our proposed model.}
Furthermore, we note that EBIC is model consistent, and therefore, assuming that the true precision matrix contains non-zero entries in each element, will choose asymptotically the clustering that has only one cluster with all variables in it.
However, as an advantage for EBIC, we exclude that clustering.
Furthermore, we note that in contrast to EBIC and AIC, the Calinski-Harabaz Index is not a model-based cluster evaluation criterion. The Calinski-Harabaz Index is an heuristic that uses as clustering criterion the ratio of the variance within and across clusters.
As such it is expected to give reasonable clustering results if the noise is considerably smaller in magnitude than the within-cluster variable partial correlations.

We remark that EBIC and AIC is not well defined if the sample covariance matrix is singular, in particular if $n < p$ or $n \approx p$. As an ad-hod remedy, which works well in practice\footnote{In particular for the mutual funds data in the next section, where the covariance matrix was bad conditioned.}, we always add $0.001$ times the identity matrix to the covariance matrix (see also \cite{ledoit2004well}).

Finally, we also compare the proposed method to two previous approaches for variable clustering: the clustered graphical lasso (CGL) as proposed in \citep{tan2015cluster},
and the Dirichlet process variable clustering (DPVC) model as proposed in \citep{palla2012nonparametric}, for which the implementation is available. DPVC models the number of clusters using a Dirichlet process. 
CGL uses for model selection the mean squared error for recovering randomly left-out elements of the covariance matrix. 
CGL uses for clustering either the single linkage clustering (SLC) or the average linkage clustering (ALC) method. 
For conciseness, we show only the results for ALC, since they tended to be better than SLC.

The results of all experiments with noise level $\eta \in \{0.0, 0.01, 0.1\}$ are shown in Tables \ref{tab:simDataNoNoiseBalanced} and \ref{tab:simDataNoiseOnPrecisionBalanced}, for balanced clusters, and Tables \ref{tab:simDataNoNoiseUnbalanced} and \ref{tab:simDataNoiseOnPrecisionUnbalanced}, for unbalanced clusters.

The tables also contain the performance of the proposed method for $\beta \in \{0, 0.01, 0.02, 0.03\}$.
Note that $\beta = 0.0$ corresponds to the basic inverse Wishart prior model for which we can calculate the marginal likelihood analytically.

Comparing the proposed method with different $\beta$, we see that $\beta = 0.02$ offers good clustering performance in the no noise and noisy setting.
In contrast, model selection with EBIC and AIC performs, as expected, well in the no noise scenario, however, in the noisy setting they tend to select incorrect clusterings.
In particular for large sample sizes EBIC tends to fail to identify correct clusterings.

The Calinski-Harabaz Index performs well in the noisy settings, whereas in the no noise setting it performs unsatisfactory.

\begin{table*}[h]
\footnotesize
  \caption{Evaluation of clustering results for $p = 40$, 
  $n \in \{20, 40, 400, 4000, 40000, 4000000\}$. Ground truth is 4 balanced clusters. 
  Shows the ANMI of the selected models (standard deviation in brackets).  No noise is added.}
  \label{tab:simDataNoNoiseBalanced}
  \begin{tabular}{lllllll}
  \toprule
    & \multicolumn{6}{c}{$\Sigma_j \sim \text{InvW}(p_j + 1, I_{p_j})$, no noise} \\
  \cmidrule{2-7}
   &  20 & 40 & 400 & 4000 & 40000 & 4000000 \\
    \midrule
Proposed ($\beta = 0.01$) &     0.76 (0.14) & 0.93 (0.09) & 1.0 (0.0) & 1.0 (0.0) & 1.0 (0.0) & 1.0 (0.0) \\ 
Proposed ($\beta = 0.02$) &     0.7 (0.2) & 0.92 (0.08) & 1.0 (0.0) & 1.0 (0.0) & 1.0 (0.0) & 1.0 (0.0) \\ 
Proposed ($\beta = 0.03$) &     0.67 (0.18) & 0.88 (0.14) & 1.0 (0.0) & 1.0 (0.0) & 1.0 (0.0) & 1.0 (0.0) \\ 
basic inverse Wishart prior         &         0.73 (0.17) & 0.93 (0.09) & 1.0 (0.0) & 1.0 (0.0) & 1.0 (0.0) & 1.0 (0.0) \\ 
EBIC ($\gamma = 0$)      &      0.12 (0.15) & 0.92 (0.08) & 1.0 (0.0) & 1.0 (0.0) & 1.0 (0.0) & 1.0 (0.0) \\ 
EBIC ($\gamma = 0.5$)            &      0.36 (0.03) & 0.51 (0.04) & 0.99 (0.03) & 1.0 (0.0) & 1.0 (0.0) & 1.0 (0.0) \\ 
EBIC ($\gamma = 1.0$)            &      0.35 (0.02) & 0.39 (0.05) & 0.96 (0.05) & 1.0 (0.0) & 1.0 (0.0) & 1.0 (0.0) \\ 
AIC &   0.12 (0.15) & 0.6 (0.49) & 1.0 (0.0) & 1.0 (0.0) & 1.0 (0.0) & 1.0 (0.0) \\ 
Calinski-Harabaz Index &        0.32 (0.03) & 0.19 (0.16) & 0.84 (0.13) & 0.73 (0.0) & 0.73 (0.0) & 0.73 (0.0) \\ 
CGL (ALC) &     0.06 (0.05) & 0.03 (0.05) & 0.11 (0.06) & 0.04 (0.04) & 0.06 (0.03) & 0.06 (0.07) \\ 
DPVC            & 	0.53 (0.07) & 0.61 (0.17) & 0.82 (0.06) & 0.93 (0.09) & NA & NA \\
\midrule
    & \multicolumn{6}{c}{$\Sigma_j \sim \text{Uniform}_{p_j}$, no noise} \\
       \midrule
Proposed ($\beta = 0.01$) &     0.12 (0.04) & 0.48 (0.07) & 0.94 (0.06) & 1.0 (0.0) & 1.0 (0.0) & 1.0 (0.0) \\ 
Proposed ($\beta = 0.02$) &     0.12 (0.05) & 0.4 (0.04) & 0.93 (0.06) & 1.0 (0.0) & 1.0 (0.0) & 1.0 (0.0) \\ 
Proposed ($\beta = 0.03$) &     0.12 (0.05) & 0.39 (0.03) & 0.93 (0.06) & 1.0 (0.0) & 1.0 (0.0) & 1.0 (0.0) \\ 
basic inverse Wishart prior         &         0.14 (0.05) & 0.76 (0.1) & 1.0 (0.0) & 1.0 (0.0) & 1.0 (0.0) & 1.0 (0.0) \\ 
EBIC ($\gamma = 0$)      &      0.07 (0.04) & 0.87 (0.09) & 1.0 (0.0) & 1.0 (0.0) & 1.0 (0.0) & 1.0 (0.0) \\ 
EBIC ($\gamma = 0.5$)            &      0.11 (0.05) & 0.48 (0.12) & 1.0 (0.0) & 1.0 (0.0) & 1.0 (0.0) & 1.0 (0.0) \\ 
EBIC ($\gamma = 1.0$)            &      0.11 (0.05) & 0.38 (0.05) & 1.0 (0.0) & 1.0 (0.0) & 1.0 (0.0) & 1.0 (0.0) \\ 
AIC &   0.07 (0.04) & 0.66 (0.34) & 1.0 (0.0) & 1.0 (0.0) & 1.0 (0.0) & 1.0 (0.0) \\ 
Calinski-Harabaz Index &        0.15 (0.05) & 0.66 (0.16) & 0.79 (0.11) & 0.46 (0.14) & 0.65 (0.23) & 0.59 (0.17) \\ 
CGL (ALC) &     0.03 (0.02) & 0.02 (0.02) & 0.37 (0.03) & 0.39 (0.0) & 0.39 (0.0) & 0.51 (0.25) \\ 
DPVC            & 	0.01 (0.02) & 0.03 (0.03) & 0.4 (0.2) & 0.51 (0.22) & NA & NA \\
\bottomrule
    \end{tabular}
 \end{table*}

\begin{table*}[h]
\footnotesize
  \caption{Evaluation of clustering results with $p = 40$, 
  $n \in \{20, 40, 400, 4000, 40000, 4000000\}$. Ground truth is 4 balanced clusters. 
  Shows the ANMI of the selected models (standard deviation in brackets). 
  Noise is added to the precision matrix.}
  \label{tab:simDataNoiseOnPrecisionBalanced}
  \begin{tabular}{lllllll}
  \toprule 
      & \multicolumn{6}{c}{$\Sigma_j \sim \text{InvW}(p_j + 1, I_{p_j}), \Sigma_{\epsilon} \sim \text{InvW}(p + 1, I_{p}), \eta = 0.01$} \\
  \cmidrule{2-7}
       &  20 & 40 & 400 & 4000 & 40000 & 4000000 \\
    \midrule
Proposed ($\beta = 0.01$) &     0.44 (0.07) & 0.86 (0.06) & 1.0 (0.0) & 1.0 (0.0) & 1.0 (0.0) & 1.0 (0.0) \\ 
Proposed ($\beta = 0.02$) &     0.41 (0.06) & 0.86 (0.06) & 1.0 (0.0) & 1.0 (0.0) & 1.0 (0.0) & 0.99 (0.03) \\ 
Proposed ($\beta = 0.03$) &     0.38 (0.06) & 0.8 (0.06) & 1.0 (0.0) & 1.0 (0.0) & 1.0 (0.0) & 0.99 (0.03) \\ 
basic inverse Wishart prior         &         0.45 (0.07) & 0.89 (0.02) & 1.0 (0.0) & 1.0 (0.0) & 0.41 (0.04) & 0.39 (0.0) \\ 
EBIC ($\gamma = 0$)      &      0.02 (0.02) & 0.82 (0.07) & 1.0 (0.0) & 1.0 (0.0) & 0.41 (0.04) & 0.39 (0.0) \\ 
EBIC ($\gamma = 0.5$)            &      0.25 (0.08) & 0.32 (0.07) & 0.98 (0.04) & 1.0 (0.0) & 0.48 (0.13) & 0.39 (0.0) \\ 
EBIC ($\gamma = 1.0$)            &      0.23 (0.07) & 0.32 (0.07) & 0.96 (0.06) & 1.0 (0.0) & 0.66 (0.14) & 0.39 (0.0) \\ 
AIC &   0.0 (0.01) & 0.54 (0.44) & 1.0 (0.0) & 0.39 (0.0) & 0.41 (0.04) & 0.39 (0.0) \\ 
Calinski-Harabaz Index &        0.26 (0.09) & 0.3 (0.16) & 0.93 (0.1) & 0.95 (0.11) & 0.89 (0.13) & 0.84 (0.13) \\ 
CGL (ALC) &     0.01 (0.02) & 0.02 (0.05) & 0.04 (0.05) & 0.03 (0.02) & 0.05 (0.06) & 0.02 (0.02) \\
DPVC            & 	0.33 (0.07) & 0.42 (0.08) & 0.59 (0.16) & 0.21 (0.18) & NA & NA \\ 
\midrule
    & \multicolumn{6}{c}{$\Sigma_j \sim \text{InvW}(p_j + 1, I_{p_j}), \Sigma_{\epsilon} \sim \text{InvW}(p + 1, I_{p}), \eta = 0.1$} \\
       \midrule
Proposed ($\beta = 0.01$) &     0.1 (0.1) & 0.4 (0.09) & 0.93 (0.1) & 0.39 (0.0) & 0.33 (0.17) & 0.29 (0.15) \\ 
Proposed ($\beta = 0.02$) &     0.13 (0.09) & 0.41 (0.07) & 0.97 (0.04) & 0.95 (0.11) & 1.0 (0.0) & 0.99 (0.03) \\ 
Proposed ($\beta = 0.03$) &     0.13 (0.09) & 0.4 (0.09) & 0.95 (0.04) & 0.99 (0.03) & 1.0 (0.0) & 0.99 (0.03) \\ 
basic inverse Wishart prior         &         0.1 (0.1) & 0.4 (0.09) & 0.93 (0.1) & 0.23 (0.19) & 0.18 (0.21) & 0.23 (0.19) \\ 
EBIC ($\gamma = 0$)      &      0.09 (0.09) & 0.29 (0.06) & 0.94 (0.05) & 0.31 (0.15) & 0.18 (0.21) & 0.23 (0.19) \\ 
EBIC ($\gamma = 0.5$)            &      0.12 (0.05) & 0.2 (0.02) & 0.87 (0.02) & 0.41 (0.04) & 0.18 (0.21) & 0.23 (0.19) \\ 
EBIC ($\gamma = 1.0$)            &      0.14 (0.06) & 0.2 (0.02) & 0.54 (0.07) & 0.86 (0.24) & 0.18 (0.21) & 0.23 (0.19) \\ 
AIC &   -0.0 (0.0) & 0.0 (0.01) & 0.09 (0.15) & 0.23 (0.19) & 0.18 (0.21) & 0.23 (0.19) \\ 
Calinski-Harabaz Index &        0.11 (0.05) & 0.15 (0.13) & 0.94 (0.05) & 0.99 (0.03) & 1.0 (0.0) & 0.99 (0.03) \\ 
CGL (ALC) &     0.02 (0.03) & 0.0 (0.01) & 0.01 (0.01) & 0.01 (0.02) & 0.0 (0.0) & 0.0 (0.0) \\ 
DPVC            & 	0.11 (0.06) & 0.16 (0.06) & 0.27 (0.06) & 0.04 (0.04) & NA & NA \\
       \midrule
    & \multicolumn{6}{c}{$\Sigma_j \sim \text{Uniform}_{p_j}, \Sigma_{\epsilon} \sim \text{Uniform}_{p}, \eta = 0.01$} \\
       \midrule
Proposed ($\beta = 0.01$) &     0.1 (0.04) & 0.45 (0.05) & 0.92 (0.06) & 0.99 (0.03) & 0.99 (0.03) & 0.93 (0.1) \\ 
Proposed ($\beta = 0.02$) &     0.12 (0.03) & 0.43 (0.06) & 0.92 (0.06) & 0.99 (0.03) & 0.99 (0.03) & 0.93 (0.1) \\ 
Proposed ($\beta = 0.03$) &     0.13 (0.02) & 0.39 (0.03) & 0.89 (0.07) & 0.99 (0.03) & 0.99 (0.03) & 0.93 (0.1) \\ 
basic inverse Wishart prior         &         0.11 (0.06) & 0.65 (0.12) & 0.94 (0.06) & 0.88 (0.12) & 0.3 (0.28) & 0.46 (0.14) \\ 
EBIC ($\gamma = 0$)      &      0.06 (0.04) & 0.78 (0.14) & 0.92 (0.1) & 0.81 (0.23) & 0.3 (0.28) & 0.46 (0.14) \\ 
EBIC ($\gamma = 0.5$)            &      0.1 (0.03) & 0.44 (0.06) & 0.94 (0.06) & 0.99 (0.03) & 0.3 (0.28) & 0.46 (0.14) \\ 
EBIC ($\gamma = 1.0$)            &      0.1 (0.03) & 0.39 (0.03) & 0.94 (0.06) & 0.99 (0.03) & 0.3 (0.28) & 0.46 (0.14) \\ 
AIC &   0.06 (0.04) & 0.24 (0.33) & 0.35 (0.43) & 0.44 (0.15) & 0.3 (0.28) & 0.46 (0.14) \\ 
Calinski-Harabaz Index &        0.14 (0.06) & 0.54 (0.33) & 0.57 (0.35) & 0.76 (0.21) & 0.59 (0.29) & 0.66 (0.14) \\ 
CGL (ALC) &     0.0 (0.01) & 0.01 (0.01) & 0.24 (0.18) & 0.39 (0.0) & 0.35 (0.08) & 0.39 (0.0) \\ 
DPVC            & 	-0.01 (0.01) & 0.06 (0.07) & 0.29 (0.22) & 0.44 (0.2) & NA & NA \\
       \midrule
      & \multicolumn{6}{c}{$\Sigma_j \sim \text{Uniform}_{p_j}, \Sigma_{\epsilon} \sim \text{Uniform}_{p}, \eta = 0.1$} \\
       \midrule
Proposed ($\beta = 0.01$) &     0.11 (0.02) & 0.45 (0.05) & 0.88 (0.07) & 0.79 (0.21) & 0.56 (0.34) & 0.64 (0.22) \\ 
Proposed ($\beta = 0.02$) &     0.14 (0.04) & 0.4 (0.02) & 0.86 (0.07) & 0.9 (0.07) & 0.56 (0.34) & 0.64 (0.22) \\ 
Proposed ($\beta = 0.03$) &     0.14 (0.04) & 0.39 (0.03) & 0.86 (0.07) & 0.9 (0.07) & 0.56 (0.34) & 0.64 (0.22) \\ 
basic inverse Wishart prior         &         0.13 (0.04) & 0.52 (0.07) & 0.88 (0.07) & 0.42 (0.33) & 0.15 (0.19) & 0.23 (0.19) \\ 
EBIC ($\gamma = 0$)      &      0.12 (0.06) & 0.7 (0.1) & 0.78 (0.22) & 0.42 (0.33) & 0.15 (0.19) & 0.16 (0.19) \\ 
EBIC ($\gamma = 0.5$)            &      0.13 (0.04) & 0.44 (0.05) & 0.88 (0.07) & 0.48 (0.26) & 0.15 (0.19) & 0.16 (0.19) \\ 
EBIC ($\gamma = 1.0$)            &      0.12 (0.05) & 0.39 (0.03) & 0.88 (0.07) & 0.6 (0.3) & 0.15 (0.19) & 0.16 (0.19) \\ 
AIC &   0.12 (0.06) & 0.2 (0.17) & 0.06 (0.12) & 0.42 (0.33) & 0.15 (0.19) & 0.16 (0.19) \\ 
Calinski-Harabaz Index &        0.17 (0.06) & 0.48 (0.29) & 0.28 (0.34) & 0.9 (0.07) & 0.49 (0.27) & 0.63 (0.22) \\ 
CGL (ALC) &     0.01 (0.01) & 0.07 (0.08) & 0.31 (0.15) & 0.39 (0.0) & 0.33 (0.11) & 0.38 (0.02) \\
DPVC            & 	-0.0 (0.0) & 0.1 (0.09) & 0.35 (0.12) & 0.19 (0.18) & NA & NA \\
       \bottomrule
    \end{tabular}
 \end{table*}

 \begin{table*}[h]
\footnotesize
  \caption{Evaluation of clustering results for $p = 40$, 
  $n \in \{20, 40, 400, 4000, 40000, 4000000\}$. Ground truth is 4 unbalanced clusters with sizes 20, 10, 5, 5. 
  Shows the ANMI of the selected models (standard deviation in brackets).  No noise is added.}
  \label{tab:simDataNoNoiseUnbalanced}
  \begin{tabular}{lllllll}
  \toprule
    & \multicolumn{6}{c}{$\Sigma_j \sim \text{InvW}(p_j + 1, I_{p_j})$, no noise} \\
  \cmidrule{2-7}
   &  20 & 40 & 400 & 4000 & 40000 & 4000000 \\
    \midrule
Proposed ($\beta = 0.01$) &     0.49 (0.15) & 0.84 (0.11) & 1.0 (0.0) & 1.0 (0.0) & 1.0 (0.0) & 1.0 (0.0) \\ 
Proposed ($\beta = 0.02$) &     0.47 (0.17) & 0.84 (0.11) & 0.99 (0.02) & 1.0 (0.0) & 1.0 (0.0) & 1.0 (0.0) \\ 
Proposed ($\beta = 0.03$) &     0.42 (0.19) & 0.82 (0.13) & 0.99 (0.02) & 1.0 (0.0) & 1.0 (0.0) & 1.0 (0.0) \\ 
basic inverse Wishart prior         &         0.5 (0.15) & 0.84 (0.12) & 1.0 (0.0) & 1.0 (0.0) & 1.0 (0.0) & 1.0 (0.0) \\ 
EBIC ($\gamma = 0$)      &      0.2 (0.17) & 0.8 (0.12) & 1.0 (0.0) & 1.0 (0.0) & 1.0 (0.0) & 1.0 (0.0) \\ 
EBIC ($\gamma = 0.5$)            &      0.24 (0.05) & 0.37 (0.05) & 0.99 (0.02) & 1.0 (0.0) & 1.0 (0.0) & 1.0 (0.0) \\ 
EBIC ($\gamma = 1.0$)            &      0.23 (0.06) & 0.32 (0.04) & 0.99 (0.02) & 1.0 (0.0) & 1.0 (0.0) & 1.0 (0.0) \\ 
AIC &   0.15 (0.19) & 0.16 (0.12) & 1.0 (0.0) & 1.0 (0.0) & 1.0 (0.0) & 1.0 (0.0) \\ 
Calinski-Harabaz Index &        0.17 (0.09) & 0.17 (0.23) & 0.46 (0.27) & 0.45 (0.23) & 0.47 (0.19) & 0.4 (0.14) \\ 
CGL (ALC) &     0.07 (0.11) & 0.03 (0.04) & 0.05 (0.07) & 0.03 (0.03) & 0.07 (0.07) & 0.05 (0.06) \\ 
DPVC            & 	0.57 (0.13) & 0.66 (0.07) & 0.64 (0.14) & 0.87 (0.17) & NA & NA \\ 
\midrule
    & \multicolumn{6}{c}{$\Sigma_j \sim \text{Uniform}_{p_j}$, no noise} \\
       \midrule
Proposed ($\beta = 0.01$) &     0.15 (0.03) & 0.33 (0.03) & 0.87 (0.1) & 0.98 (0.03) & 1.0 (0.0) & 0.98 (0.03) \\ 
Proposed ($\beta = 0.02$) &     0.15 (0.03) & 0.33 (0.03) & 0.87 (0.1) & 0.97 (0.04) & 1.0 (0.0) & 0.97 (0.04) \\ 
Proposed ($\beta = 0.03$) &     0.16 (0.03) & 0.31 (0.03) & 0.67 (0.18) & 0.97 (0.04) & 0.98 (0.03) & 0.97 (0.04) \\
basic inverse Wishart prior         &         0.17 (0.05) & 0.33 (0.02) & 1.0 (0.0) & 1.0 (0.0) & 1.0 (0.0) & 1.0 (0.0) \\ 
EBIC ($\gamma = 0$)      &      0.08 (0.09) & 0.6 (0.23) & 1.0 (0.0) & 1.0 (0.0) & 1.0 (0.0) & 1.0 (0.0) \\ 
EBIC ($\gamma = 0.5$)            &      0.16 (0.03) & 0.33 (0.04) & 0.98 (0.03) & 1.0 (0.0) & 1.0 (0.0) & 1.0 (0.0) \\ 
EBIC ($\gamma = 1.0$)            &      0.16 (0.03) & 0.31 (0.03) & 0.91 (0.12) & 1.0 (0.0) & 1.0 (0.0) & 1.0 (0.0) \\ 
AIC &   0.08 (0.08) & 0.52 (0.33) & 1.0 (0.0) & 1.0 (0.0) & 1.0 (0.0) & 1.0 (0.0) \\ 
Calinski-Harabaz Index &        0.16 (0.06) & 0.53 (0.3) & 0.64 (0.15) & 0.63 (0.28) & 0.88 (0.17) & 0.96 (0.08) \\ 
CGL (ALC) &     -0.01 (0.01) & -0.01 (0.0) & -0.0 (0.01) & 0.15 (0.16) & 0.15 (0.21) & 0.12 (0.06) \\ 
DPVC            & 	0.02 (0.01) & 0.0 (0.04) & 0.23 (0.14) & 0.25 (0.13) & NA & NA \\ 
       \bottomrule
    \end{tabular}
 \end{table*}

 \begin{table*}[h]
\footnotesize
  \caption{Evaluation of clustering results with $p = 40$, 
  $n \in \{20, 40, 400, 4000, 40000, 4000000\}$. Ground truth is 4 unbalanced clusters with sizes 20, 10, 5, 5. 
  Shows the ANMI of the selected models (standard deviation in brackets). 
  Noise is added to the precision matrix.}
  \label{tab:simDataNoiseOnPrecisionUnbalanced}
  \begin{tabular}{lllllll}
  \toprule 
      & \multicolumn{6}{c}{$\Sigma_j \sim \text{InvW}(p_j + 1, I_{p_j}), \Sigma_{\epsilon} \sim \text{InvW}(p + 1, I_{p}), \eta = 0.01$} \\
  \cmidrule{2-7}
       &  20 & 40 & 400 & 4000 & 40000 & 4000000 \\
    \midrule
Proposed ($\beta = 0.01$) &     0.45 (0.14) & 0.75 (0.15) & 1.0 (0.0) & 1.0 (0.0) & 1.0 (0.0) & 1.0 (0.0) \\ 
Proposed ($\beta = 0.02$) &     0.39 (0.09) & 0.75 (0.15) & 1.0 (0.0) & 1.0 (0.0) & 1.0 (0.0) & 0.98 (0.03) \\ 
Proposed ($\beta = 0.03$) &     0.39 (0.09) & 0.7 (0.18) & 1.0 (0.0) & 0.97 (0.06) & 1.0 (0.0) & 0.98 (0.03) \\ 
basic inverse Wishart prior         &         0.48 (0.15) & 0.8 (0.09) & 1.0 (0.0) & 0.91 (0.11) & 0.39 (0.13) & 0.42 (0.12) \\ 
EBIC ($\gamma = 0$)      &      0.12 (0.08) & 0.67 (0.12) & 1.0 (0.0) & 0.91 (0.11) & 0.48 (0.17) & 0.42 (0.12) \\ 
EBIC ($\gamma = 0.5$)            &      0.19 (0.08) & 0.32 (0.04) & 0.97 (0.03) & 1.0 (0.0) & 0.54 (0.26) & 0.42 (0.12) \\ 
EBIC ($\gamma = 1.0$)            &      0.17 (0.07) & 0.28 (0.07) & 0.96 (0.03) & 1.0 (0.0) & 0.68 (0.24) & 0.42 (0.12) \\ 
AIC &   0.06 (0.09) & 0.3 (0.34) & 1.0 (0.0) & 0.4 (0.1) & 0.39 (0.13) & 0.42 (0.12) \\ 
Calinski-Harabaz Index &        0.2 (0.06) & 0.13 (0.2) & 0.45 (0.27) & 0.59 (0.17) & 0.7 (0.21) & 0.77 (0.03) \\ 
CGL (ALC) &     0.08 (0.06) & 0.05 (0.03) & 0.04 (0.03) & 0.03 (0.02) & 0.03 (0.02) & 0.04 (0.04) \\ 
DPVC            & 	0.28 (0.04) & 0.35 (0.07) & 0.57 (0.08) & 0.4 (0.12) & NA & NA \\ 
\midrule
    & \multicolumn{6}{c}{$\Sigma_j \sim \text{InvW}(p_j + 1, I_{p_j}), \Sigma_{\epsilon} \sim \text{InvW}(p + 1, I_{p}), \eta = 0.1$} \\
       \midrule
Proposed ($\beta = 0.01$) &     0.09 (0.11) & 0.42 (0.12) & 0.84 (0.1) & 0.42 (0.16) & 0.18 (0.22) & 0.24 (0.18) \\ 
Proposed ($\beta = 0.02$) &     0.09 (0.11) & 0.42 (0.13) & 0.88 (0.11) & 0.85 (0.15) & 0.99 (0.02) & 0.9 (0.09) \\ 
Proposed ($\beta = 0.03$) &     0.15 (0.06) & 0.42 (0.13) & 0.89 (0.09) & 0.92 (0.07) & 0.99 (0.02) & 0.9 (0.09) \\ 
basic inverse Wishart prior         &         0.11 (0.14) & 0.42 (0.13) & 0.84 (0.1) & 0.2 (0.2) & 0.0 (0.01) & 0.1 (0.17) \\ 
EBIC ($\gamma = 0$)      &      0.04 (0.05) & 0.24 (0.06) & 0.88 (0.11) & 0.2 (0.2) & 0.0 (0.01) & 0.1 (0.17) \\ 
EBIC ($\gamma = 0.5$)            &      0.05 (0.02) & 0.19 (0.04) & 0.74 (0.19) & 0.44 (0.17) & 0.0 (0.01) & 0.1 (0.17) \\ 
EBIC ($\gamma = 1.0$)            &      0.05 (0.02) & 0.19 (0.04) & 0.41 (0.06) & 0.78 (0.12) & 0.0 (0.01) & 0.1 (0.17) \\ 
AIC &   -0.01 (0.01) & 0.15 (0.21) & 0.19 (0.2) & 0.2 (0.2) & 0.0 (0.01) & 0.1 (0.17) \\ 
Calinski-Harabaz Index &        0.06 (0.03) & 0.17 (0.11) & 0.68 (0.25) & 0.67 (0.2) & 0.83 (0.17) & 0.76 (0.04) \\ 
CGL (ALC) &     0.04 (0.04) & 0.03 (0.02) & 0.05 (0.06) & 0.1 (0.11) & 0.05 (0.07) & 0.08 (0.09) \\ 
DPVC            & 	0.13 (0.05) & 0.16 (0.05) & 0.3 (0.13) & 0.07 (0.03) & NA & NA \\ 
       \midrule
    & \multicolumn{6}{c}{$\Sigma_j \sim \text{Uniform}_{p_j}, \Sigma_{\epsilon} \sim \text{Uniform}_{p}, \eta = 0.01$} \\
       \midrule
Proposed ($\beta = 0.01$) &     0.11 (0.02) & 0.32 (0.04) & 0.74 (0.15) & 0.83 (0.1) & 0.59 (0.32) & 0.5 (0.33) \\ 
Proposed ($\beta = 0.02$) &     0.11 (0.02) & 0.32 (0.04) & 0.61 (0.17) & 0.83 (0.1) & 0.59 (0.32) & 0.59 (0.32) \\
Proposed ($\beta = 0.03$) &     0.11 (0.02) & 0.32 (0.04) & 0.43 (0.06) & 0.83 (0.1) & 0.59 (0.32) & 0.59 (0.32) \\ 
basic inverse Wishart prior         &         0.11 (0.02) & 0.32 (0.04) & 0.84 (0.05) & 0.28 (0.0) & 0.11 (0.14) & 0.17 (0.23) \\ 
EBIC ($\gamma = 0$)      &      0.18 (0.13) & 0.43 (0.05) & 0.76 (0.13) & 0.22 (0.12) & 0.11 (0.14) & 0.06 (0.11) \\ 
EBIC ($\gamma = 0.5$)            &      0.11 (0.02) & 0.32 (0.04) & 0.84 (0.05) & 0.51 (0.3) & 0.11 (0.14) & 0.06 (0.11) \\ 
EBIC ($\gamma = 1.0$)            &      0.11 (0.02) & 0.32 (0.04) & 0.79 (0.13) & 0.67 (0.24) & 0.11 (0.14) & 0.06 (0.11) \\ 
AIC &   0.14 (0.05) & 0.16 (0.28) & 0.17 (0.23) & 0.22 (0.12) & 0.09 (0.12) & 0.06 (0.11) \\ 
Calinski-Harabaz Index &        0.14 (0.08) & 0.32 (0.3) & 0.34 (0.33) & 0.68 (0.22) & 0.25 (0.27) & 0.41 (0.32) \\ 
CGL (ALC) &     -0.01 (0.0) & -0.01 (0.0) & 0.01 (0.04) & -0.01 (0.01) & 0.02 (0.02) & 0.01 (0.01) \\ 
DPVC            & 	0.01 (0.01) & 0.03 (0.06) & 0.2 (0.05) & 0.01 (0.02) & NA & NA \\ 
       \midrule
      & \multicolumn{6}{c}{$\Sigma_j \sim \text{Uniform}_{p_j}, \Sigma_{\epsilon} \sim \text{Uniform}_{p}, \eta = 0.1$} \\
       \midrule
Proposed ($\beta = 0.01$) &     0.1 (0.02) & 0.34 (0.07) & 0.68 (0.18) & 0.6 (0.31) & 0.09 (0.12) & 0.06 (0.11) \\ 
Proposed ($\beta = 0.02$) &     0.11 (0.02) & 0.34 (0.07) & 0.65 (0.21) & 0.7 (0.13) & 0.21 (0.21) & 0.28 (0.26) \\ 
Proposed ($\beta = 0.03$) &     0.11 (0.02) & 0.32 (0.06) & 0.58 (0.2) & 0.7 (0.13) & 0.32 (0.22) & 0.28 (0.26) \\ 
basic inverse Wishart prior         &         0.14 (0.03) & 0.37 (0.08) & 0.78 (0.1) & 0.0 (0.02) & 0.09 (0.12) & 0.06 (0.11) \\ 
EBIC ($\gamma = 0$)      &      0.16 (0.05) & 0.49 (0.21) & 0.71 (0.14) & 0.0 (0.02) & 0.09 (0.12) & 0.06 (0.11) \\ 
EBIC ($\gamma = 0.5$)            &      0.11 (0.01) & 0.36 (0.08) & 0.77 (0.13) & 0.06 (0.11) & 0.09 (0.12) & 0.06 (0.11) \\ 
EBIC ($\gamma = 1.0$)            &      0.11 (0.01) & 0.31 (0.05) & 0.7 (0.16) & 0.12 (0.14) & 0.09 (0.12) & 0.06 (0.11) \\ 
AIC &   0.15 (0.05) & 0.05 (0.12) & 0.06 (0.11) & 0.0 (0.02) & 0.09 (0.12) & 0.06 (0.11) \\ 
Calinski-Harabaz Index &        0.16 (0.05) & 0.29 (0.26) & 0.42 (0.23) & 0.45 (0.38) & 0.09 (0.12) & 0.33 (0.31) \\ 
CGL (ALC) &     -0.0 (0.01) & -0.0 (0.01) & -0.01 (0.0) & -0.01 (0.0) & -0.01 (0.0) & -0.0 (0.01) \\ 
DPVC            & 	0.0 (0.04) & 0.03 (0.05) & 0.11 (0.13) & 0.02 (0.03) & NA & NA \\ 
       \bottomrule
    \end{tabular}
 \end{table*}

In Figures \ref{fig:clusterDistributionNoNoise} and \ref{fig:clusterDistributionWithNoise},  we show the posterior distribution with and without noise on the precision matrix, respectively.\footnote{Same setting as before, $p = 40$, $\Sigma_j \sim \text{InvW}(p_j + 1, I_{p_j})$. Noise is $\Sigma_{\epsilon} \sim \text{InvW}(p + 1, I_{p}), \eta = 0.01$. Proposed method $\beta = 0.02$.}
In both cases, given that the sample size $n$ is large enough, the proposed method is able to estimate correctly the number of clusters. 
In contrast, the basic inverse Wishart prior model underestimates the number of clusters for large $n$ and existence of noise in the precision matrix.

\begin{figure*}[h]
  \centering
  \includegraphics[scale=0.4, trim=10cm 0cm 10cm 0cm]{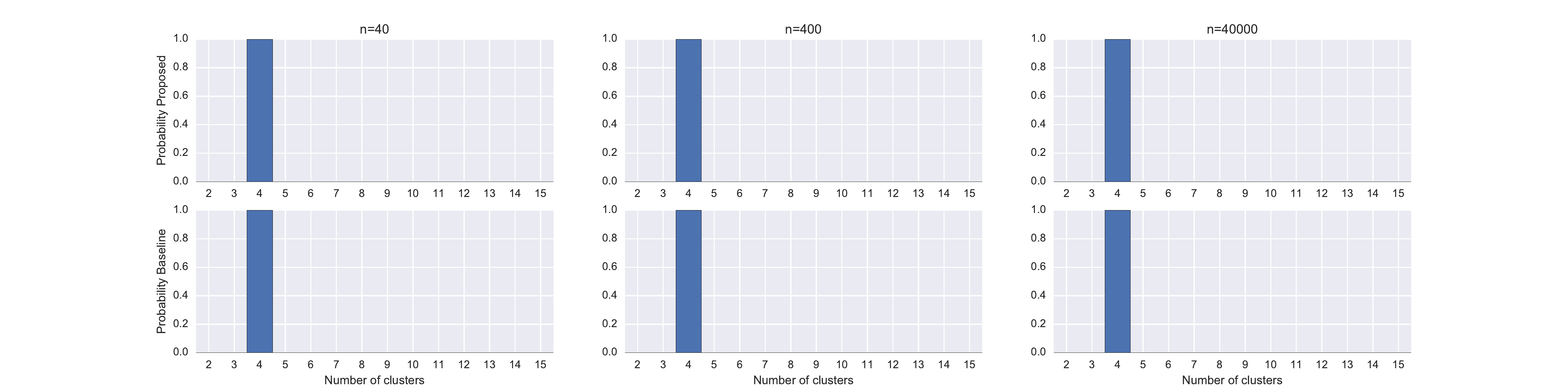}  
    \caption{Posterior distribution of the number of clusters of the proposed method (top row) and the basic inverse Wishart prior model (bottom row). Ground truth is 4 clusters; there is no noise on the precision matrix. }
  \label{fig:clusterDistributionNoNoise}
\end{figure*} 

 \begin{figure*}[h]
  \centering
  \includegraphics[scale=0.4, trim=10cm 0cm 10cm 0cm]{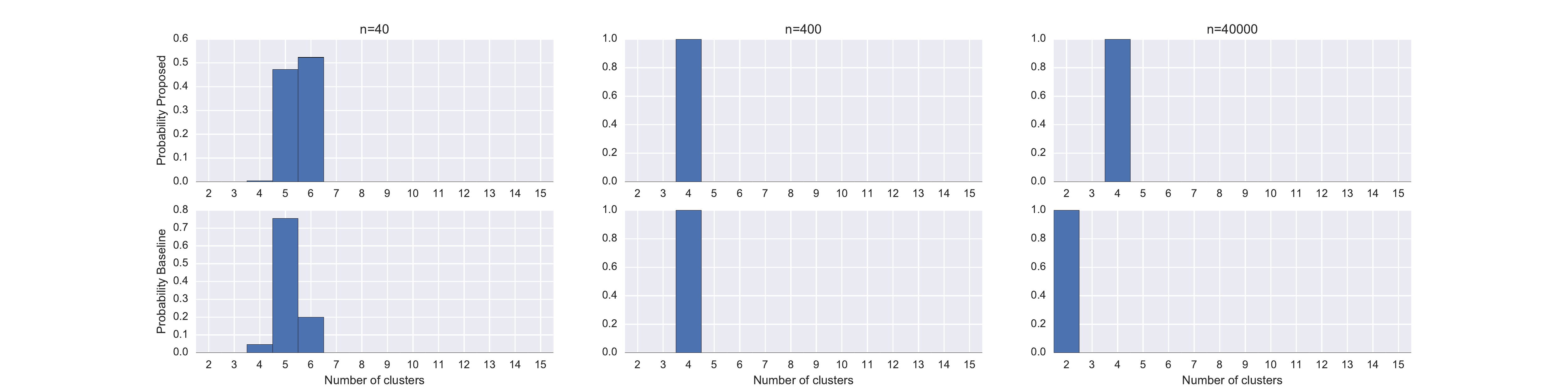}  
    \caption{Posterior distribution of the number of clusters of the proposed method (top row) and the basic inverse Wishart prior model (bottom row). Ground truth is 4 clusters; noise was added to the precision matrix. }
  \label{fig:clusterDistributionWithNoise}
\end{figure*}

\subsection{Comparison of variational and MCMC estimate}
Here, we compare our variational approximation with MCMC on a small scale simulated problem where it is computationally feasible to estimate the marginal likelihood with MCMC.
We generated synthetic data as in the previous section, only with the difference that we set the number of variables $p$ to 12. 

The number of samples $M$ for MCMC was set to 10000, where we used 10\% as burn in. 
For two randomly picked clusterings for $n = 12$, and $n = 1200000$, 
we checked the acceptance rates and convergence using the multivariate extension of the Gelman-Rubin diagnostic \citep{brooks1998general}.
The average acceptance rates were around $80\%$ and the potential scale reduction factor was 1.01.

The runtime of MCMC was around 40 minutes for evaluating one clustering, whereas for the variational approximation the runtime was around 2 seconds.\footnote{Runtime on one core of Intel(R) Xeon(R) CPU 2.30GHz.}
The results are shown in Table \ref{tab:simDataSmallComparisonVariationalMCMC}, suggesting that the quality of the selected clusterings using the variational approximation is similar to MCMC.

\begin{table*}[h]
\footnotesize
  \caption{Comparison of variational and MCMC estimate. Evaluation of clustering results for $p = 12$, 
  $n \in \{12, 120, 1200, 1200000\}$. Ground truth is 4 balanced clusters. $\beta = 0.02$.
  Shows the ANMI of the selected models (standard deviation in brackets). }
  \label{tab:simDataSmallComparisonVariationalMCMC}
  \begin{tabular}{lllll}
  \toprule
    & \multicolumn{4}{c}{$\Sigma_j \sim \text{InvW}(p_j + 1, I_{p_j})$, no noise} \\
  \cmidrule{2-5}
   &   12 & 120 & 1200 & 1200000\\
    \midrule
Proposed, variational &     0.39 (0.23) & 0.89 (0.09) & 0.96 (0.07) & 0.82 (0.11) \\ 
Proposed, MCMC & 0.37 (0.23) & 0.89 (0.09) & 0.96 (0.07) & 0.9 (0.14) \\
 basic inverse Wishart prior & 	0.39 (0.23) & 0.89 (0.09) & 1.0 (0.0) & 1.0 (0.0) \\  
\midrule
    & \multicolumn{4}{c}{$\Sigma_j \sim \text{Uniform}_{p_j}$, no noise} \\
  \cmidrule{2-5}
    Proposed, variational &     0.76 (0.17) & 1.0 (0.0) & 1.0 (0.0) & 1.0 (0.0) \\ 
    Proposed, MCMC & 0.66 (0.1) & 1.0 (0.0) & 1.0 (0.0) & 1.0 (0.0) \\ 
 basic inverse Wishart prior &	0.76 (0.17) & 1.0 (0.0) & 1.0 (0.0) & 1.0 (0.0)  \\
\midrule
& \multicolumn{4}{c}{$\Sigma_j \sim \text{InvW}(p_j + 1, I_{p_j}), \Sigma_{\epsilon} \sim \text{InvW}(p + 1, I_{p}), \eta = 0.01$} \\
  \cmidrule{2-5}
  Proposed, variational &     0.42 (0.27) & 0.8 (0.16) & 1.0 (0.0) & 0.96 (0.07) \\ 
  Proposed, MCMC & 0.17 (0.24) & 0.8 (0.16) & 1.0 (0.0) & 0.96 (0.07) \\ 
 basic inverse Wishart prior & 	0.42 (0.27) & 0.94 (0.12) & 0.93 (0.13) & 0.34 (0.04)  \\
\midrule
& \multicolumn{4}{c}{$\Sigma_j \sim \text{InvW}(p_j + 1, I_{p_j}), \Sigma_{\epsilon} \sim \text{InvW}(p + 1, I_{p}), \eta = 0.1$} \\
  \cmidrule{2-5}
Proposed, variational &     0.11 (0.16) & 0.57 (0.07) & 0.55 (0.26) & 0.78 (0.2) \\ 
Proposed, MCMC & 0.09 (0.06) & 0.61 (0.13) & 0.61 (0.23) & 0.78 (0.2) \\ 
 basic inverse Wishart prior & 	0.16 (0.15) & 0.54 (0.1) & 0.28 (0.15) & 0.21 (0.18) \\
 \midrule
   & \multicolumn{4}{c}{$\Sigma_j \sim \text{Uniform}_{p_j}, \Sigma_{\epsilon} \sim \text{Uniform}_{p}, \eta = 0.01$} \\
  \cmidrule{2-5}
Proposed, variational &     0.79 (0.12) & 0.82 (0.26) & 0.73 (0.33) & 0.96 (0.07) \\ 
Proposed, MCMC &   0.82 (0.11) & 0.96 (0.09) & 0.75 (0.31) & 0.96 (0.07) \\
 basic inverse Wishart prior & 0.79 (0.12) & 0.48 (0.15) & 0.28 (0.09) & 0.28 (0.09)  \\
 \midrule
 & \multicolumn{4}{c}{$\Sigma_j \sim \text{Uniform}_{p_j}, \Sigma_{\epsilon} \sim \text{Uniform}_{p}, \eta = 0.1$} \\
  \cmidrule{2-5}
  Proposed, variational &     0.67 (0.22) & 0.24 (0.24) & 0.32 (0.0) & 0.35 (0.18) \\
  Proposed, MCMC & 0.68 (0.17) & 0.24 (0.24) & 0.46 (0.27) & 0.35 (0.18) \\
 basic inverse Wishart prior & 0.69 (0.21) & 0.13 (0.11) & 0.26 (0.13) & 0.28 (0.09) \\ 
         \bottomrule
    \end{tabular}
 \end{table*}

\section{Real Data Experiments} \label{sec:realDataExperiments}

In this section, we investigate the properties of the proposed model selection criterion on three real data sets.
In all cases, we use the spectral clustering algorithm from Appendix \ref{app:spectralClustering} to create cluster candidates.
All variables were normalized to have mean 0 and variance 1.
For all methods, except DPVC, the number of clusters is considered to be in $\{2, 3, 4, \ldots, \min(p - 1,15) \}$. 
DPVC automatically selects the number of clusters by assuming a Dirichlet process prior.
We evaluated the proposed method with $\beta = 0.02$ using the variational approximation.

\subsection{Mutual Funds}

Here we use the mutual funds data, which has been previously analyzed in \citep{scott2008feature,marlin2009group}.
The data contains 59 mutual funds (p = 59) grouped into 4 clusters:
U.S. bond funds, U.S. stock funds, balanced funds (containing U.S. stocks and bonds), and international stock funds. 
The number of observations is 86.

The results of all methods are visualized in Table \ref{tab:mutualFundResults}. 
It is difficult to interpret the results produced by EBIC ($\gamma = 1.0$), AIC and the Calinski-Harabaz Index.
In contrast, the proposed method and EBIC ($\gamma = 0.0$) produce results that are easier to interpret.
In particular, our results suggest that there is a considerable correlation between the balanced funds and the U.S. stock funds which was also observed in \cite{marlin2009group}.

In Figure \ref{fig:mutualFundsDataVisualized} we show a two dimensional representation of the data, that was found using Laplacian Eigenmaps \citep{belkin2003laplacian}.
The figure supports the claim that balanced funds and the U.S. stock funds have similar behavior. 

\begin{figure*}[h]
  \centering
  \includegraphics[scale=0.5]{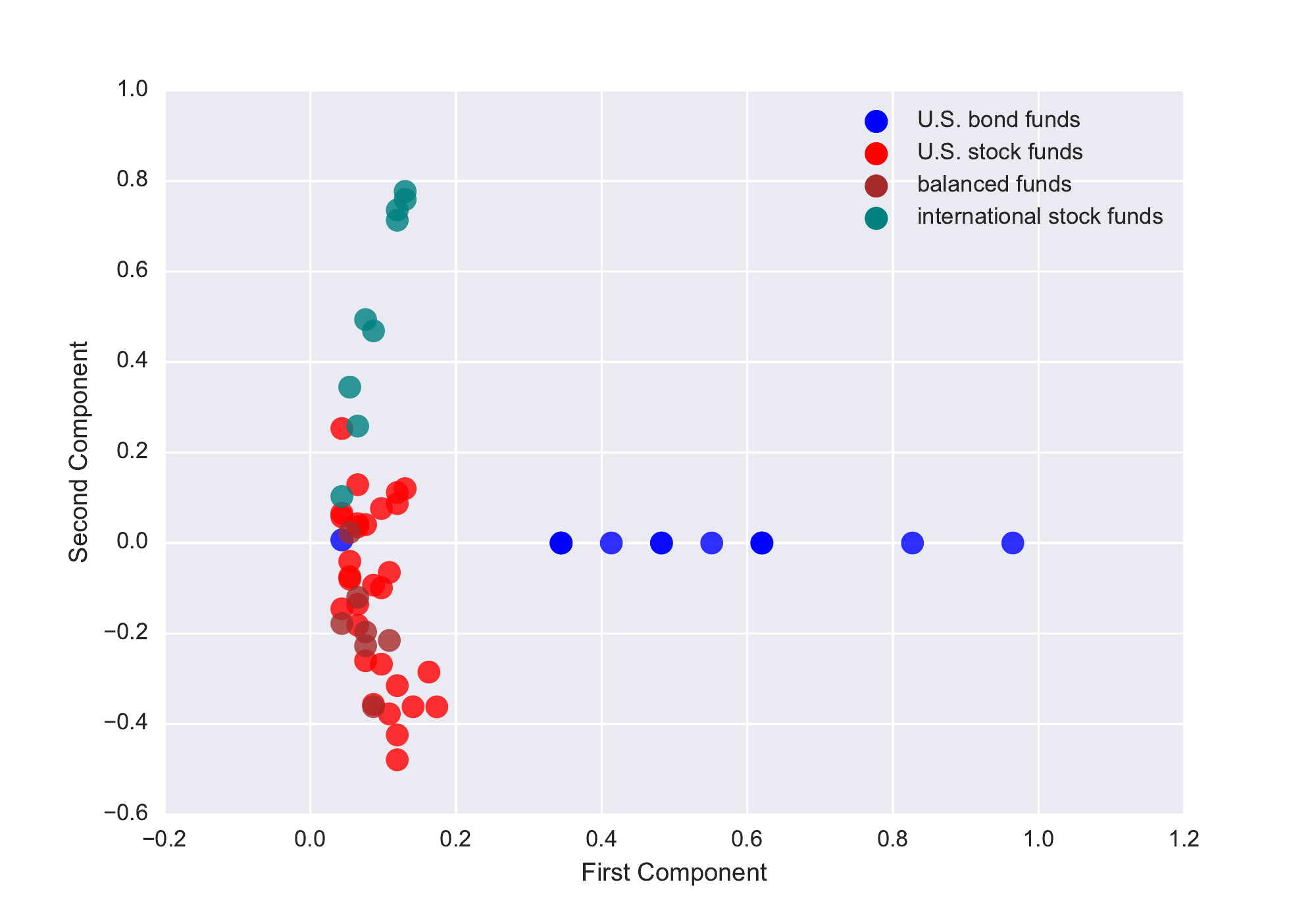} 
    \caption{Two dimensional representation of the mutual funds data. }
  \label{fig:mutualFundsDataVisualized}
\end{figure*}

 \begin{table*}[h]
  \caption{Evaluation of selected clusterings of the mutual funds data. Colors highlight the type of fund. Numbers denote the cluster id assigned by the respective method.
  Here the size of the restricted hypotheses space $|\mathscr{C}^*|$ found by spectral clustering was 128.} 
  \label{tab:mutualFundResults}
 \center
 \small
 \begin{tabular}{ll}
   \toprule
 \multicolumn{2}{c}{Proposed and EBIC ($\gamma = 0.0$) [number of clusters = 6, ANMI = 0.48]} \\
   \midrule
\color{blue}{U.S. bond funds} &  \textbf{ \color{blue}{ 2 2 2 2 2 2 2 4 2 2 2 2 2  } } \\
\color{red}{U.S. stock funds} &  \textbf{ \color{red}{ 1 1 1 1 1 1 1 1 1 1 1 1 1 1 1 1 1 1 1 1 1 1 1 1 1 1 5 1 4 6  } } \\
\color{brown}{balanced funds} &  \textbf{ \color{brown}{ 1 1 1 1 1 1 1  } } \\
\color{teal}{international stock funds} &  \textbf{ \color{teal}{ 1 3 1 1 3 1 3 3 1  } } \\
    \midrule
 \multicolumn{2}{c}{basic inverse Wishart prior [number of clusters = 3, ANMI = 0.42]} \\
   \midrule
\color{blue}{U.S. bond funds} &  \textbf{ \color{blue}{ 2 2 2 2 2 2 2 2 2 2 2 2 2  } } \\
\color{red}{U.S. stock funds} &  \textbf{ \color{red}{ 1 1 1 1 1 1 1 1 1 1 1 1 1 1 1 1 1 1 1 1 1 1 1 1 1 1 3 1 1 1  } } \\
\color{brown}{balanced funds} &  \textbf{ \color{brown}{ 1 1 1 1 1 1 1  } } \\
\color{teal}{international stock funds} &  \textbf{ \color{teal}{ 1 1 1 1 1 1 1 1 1  } } \\
    \midrule
 \multicolumn{2}{c}{EBIC ($\gamma = 0.5$) [number of clusters = 11, ANMI = 0.32] } \\
   \midrule
\color{blue}{U.S. bond funds} &  \textbf{ \color{blue}{ 2 9 2 9 2 2 2 1 10 9 2 2 2  } } \\
\color{red}{U.S. stock funds} &  \textbf{ \color{red}{ 7 11 7 11 7 11 7 7 11 5 7 11 5 1 8 7 11 5 5 5 5 5 5 5 8 5 4 8 8 6  } } \\
\color{brown}{balanced funds} &  \textbf{ \color{brown}{ 11 7 8 7 11 7 11  } } \\
\color{teal}{international stock funds} &  \textbf{ \color{teal}{ 1 3 1 1 3 1 3 3 3  } } \\
   \midrule
 \multicolumn{2}{c}{EBIC ($\gamma = 1.0$) [number of clusters = 14, ANMI = 0.25] } \\
   \midrule
\color{blue}{U.S. bond funds} &  \textbf{ \color{blue}{ 2 9 2 9 2 14 2 1 14 9 10 10 10  } } \\
\color{red}{U.S. stock funds} &  \textbf{ \color{red}{ 12 8 12 6 12 8 12 12 8 6 12 8 6 3 11 6 8 5 7 5 5 5 5 6 11 5 11 15 4 11  } } \\
\color{brown}{balanced funds} &  \textbf{ \color{brown}{ 8 12 1 12 8 6 7  } } \\
\color{teal}{international stock funds} &  \textbf{ \color{teal}{ 3 13 3 3 13 3 13 13 13  } } \\
  \midrule
 \multicolumn{2}{c}{AIC and Calinski-Harabaz Index [number of clusters = 2, ANMI = 0]} \\
    \midrule   
\color{blue}{U.S. bond funds} &  \textbf{ \color{blue}{ 1 1 1 1 1 1 1 1 1 1 1 1 1  } } \\
\color{red}{U.S. stock funds} &  \textbf{ \color{red}{ 1 1 1 1 1 1 1 1 1 1 1 1 1 1 1 1 1 1 1 1 1 1 1 1 1 1 2 1 1 1  } } \\
\color{brown}{balanced funds} &  \textbf{ \color{brown}{ 1 1 1 1 1 1 1  } } \\
\color{teal}{international stock funds} &  \textbf{ \color{teal}{ 1 1 1 1 1 1 1 1 1  } } \\
\midrule
 \multicolumn{2}{c}{CGL (ALC) [number of clusters = 3, ANMI = 0.36]} \\
 \midrule   
 \color{blue}{U.S. bond funds} &  \textbf{ \color{blue}{ 1 1 1 1 1 1 1 3 1 1 1 1 1  } } \\
\color{red}{U.S. stock funds} &  \textbf{ \color{red}{ 2 2 2 2 2 2 2 2 2 2 2 2 2 2 2 2 2 2 2 2 2 2 2 2 2 2 3 3 3 3  } } \\
\color{brown}{balanced funds} &  \textbf{ \color{brown}{ 2 2 2 2 3 2 2  } } \\
\color{teal}{international stock funds} &  \textbf{ \color{teal}{ 2 2 2 2 2 2 2 3 2  } } \\
 \midrule
 \multicolumn{2}{c}{DPVC [number of clusters = 2, ANMI = 0.35]} \\
 \midrule   
\color{blue}{U.S. bond funds} &  \textbf{ \color{blue}{ 1 1 1 1 1 1 1 2 1 1 1 1 1  } } \\
\color{red}{U.S. stock funds} &  \textbf{ \color{red}{ 2 2 2 2 2 2 2 2 2 2 2 2 2 2 2 2 2 2 2 2 2 2 2 2 2 2 2 2 2 2  } } \\
\color{brown}{balanced funds} &  \textbf{ \color{brown}{ 2 2 2 2 2 2 2  } } \\
\color{teal}{international stock funds} &  \textbf{ \color{teal}{ 2 2 2 2 2 2 2 2 2  } } \\
\bottomrule
 \end{tabular}
\end{table*}

\subsection{Gene Regulations}

We tested our method also on the gene expression data that was analyzed in \citep{hirose2017robust}. 
The data consists of 11 genes with 445 gene expressions. The true gene regularizations are known in this case and shown in Figure \ref{fig:gene_regularization_true}, adapted from \citep{hirose2017robust}. 
The most important fact is that there are two independent groups of genes and any clustering that mixes these two can be considered as wrong.

We show the results of all methods in Figure \ref{fig:gene_regularization_clustering_result}, where we mark each cluster with a different color superimposed on the true regularization structure.
Here only the clustering selected by the proposed method, EBIC ($\gamma = 1.0$) and Calinski-Harabaz correctly divide the two group of genes.

\begin{figure*}[h]
  \centering
  \includegraphics[page=1,scale=0.25]{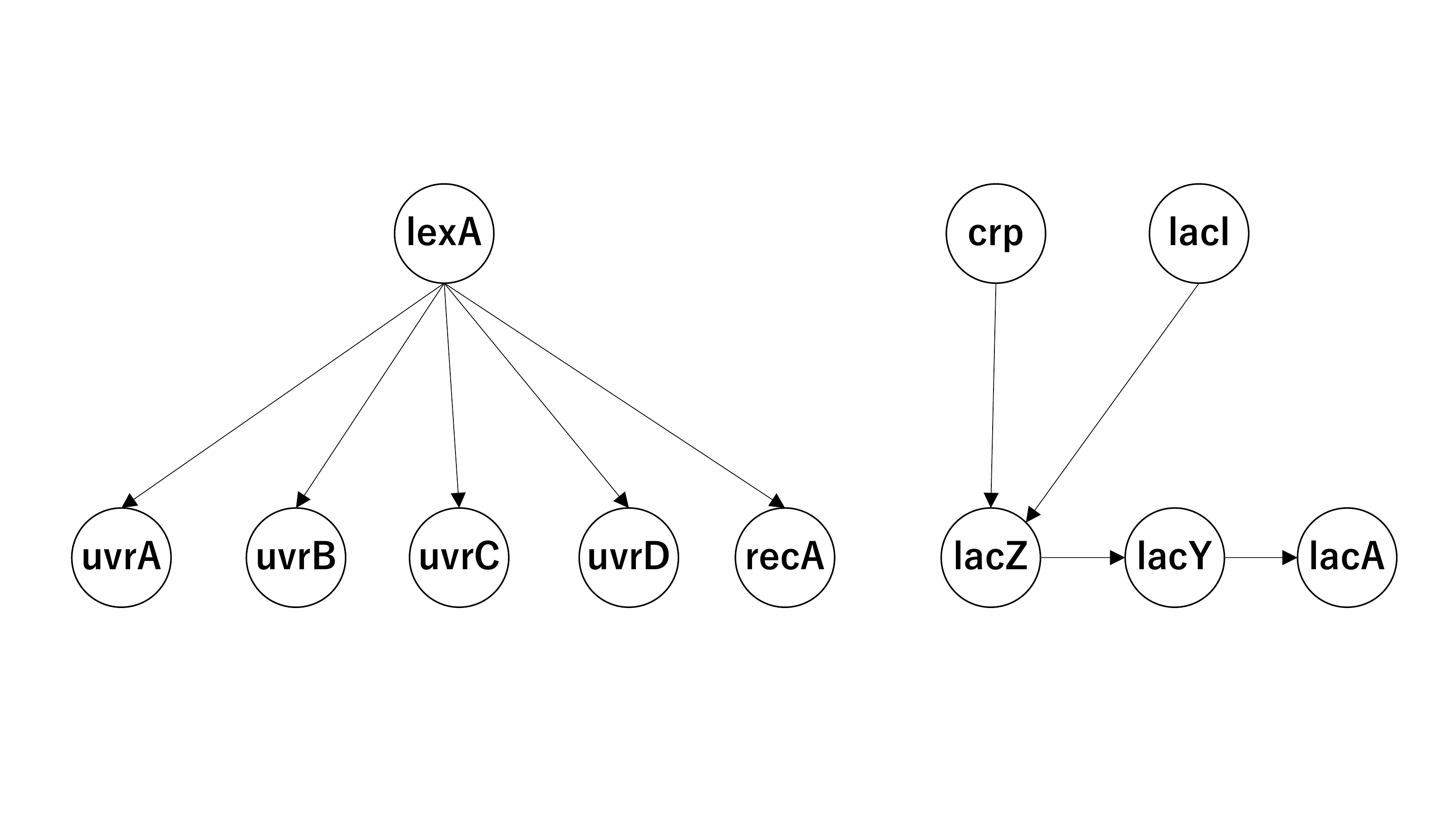} 
   \caption{Gene regulations of E. coli. as given in \citep{hirose2017robust,alberts2014molecular}}
  \label{fig:gene_regularization_true}
\end{figure*}

\begin{figure*}[h]
\centering
\begin{tabular}{*{2}{m{0.4\textwidth}}}
Proposed and EBIC ($\gamma = 1.0$) & \includegraphics[page=2,scale=0.25, clip=true, trim=0cm 4cm 0cm 3cm]{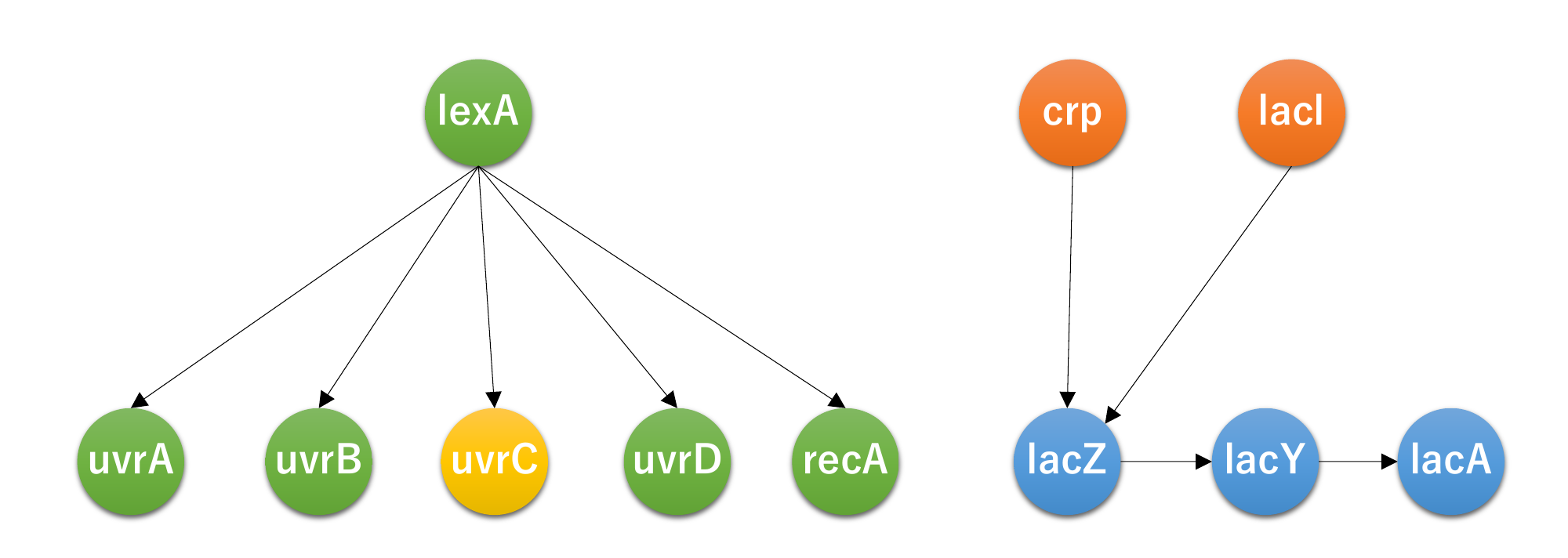} \\
 EBIC ($\gamma = 0.0$), basic inverse Wishart prior, AIC, CGL & \includegraphics[page=3,scale=0.25, clip=true, trim=0cm 4cm 0cm 3cm]{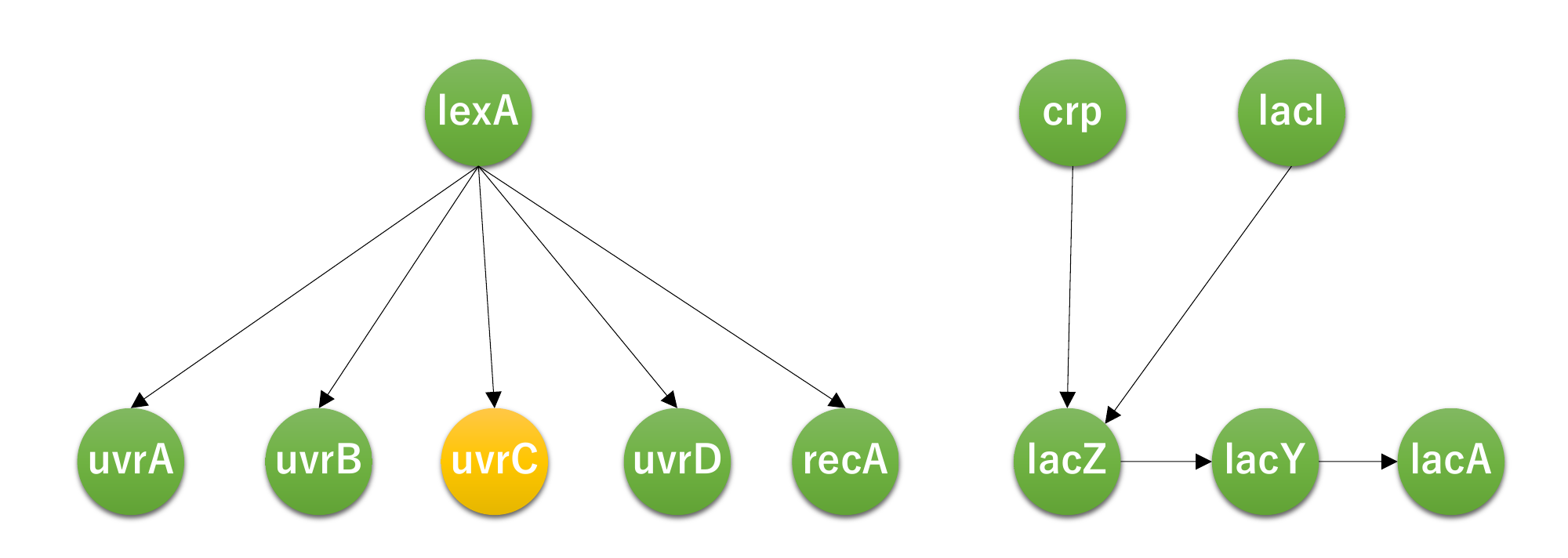} \\
EBIC ($\gamma = 0.5$) & \includegraphics[page=4,scale=0.25, clip=true, trim=0cm 4cm 0cm 3cm]{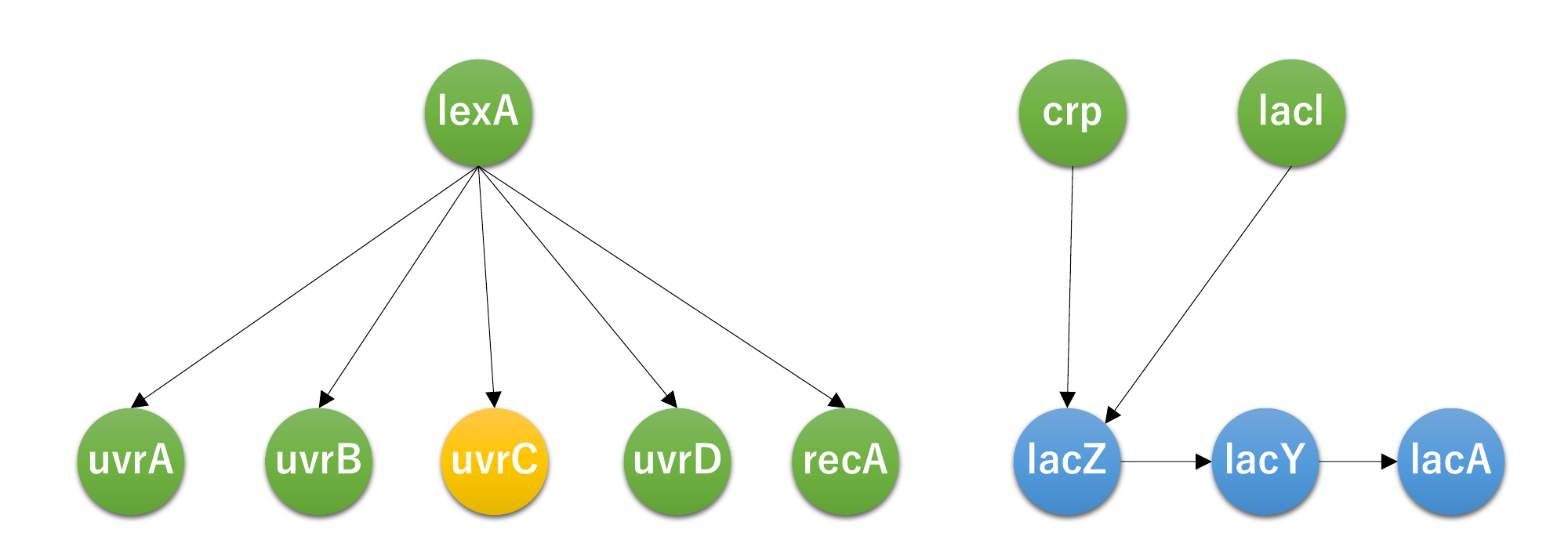} \\
Calinski-Harabaz Index & \includegraphics[page=5,scale=0.25, clip=true, trim=0cm 4cm 0cm 3cm]{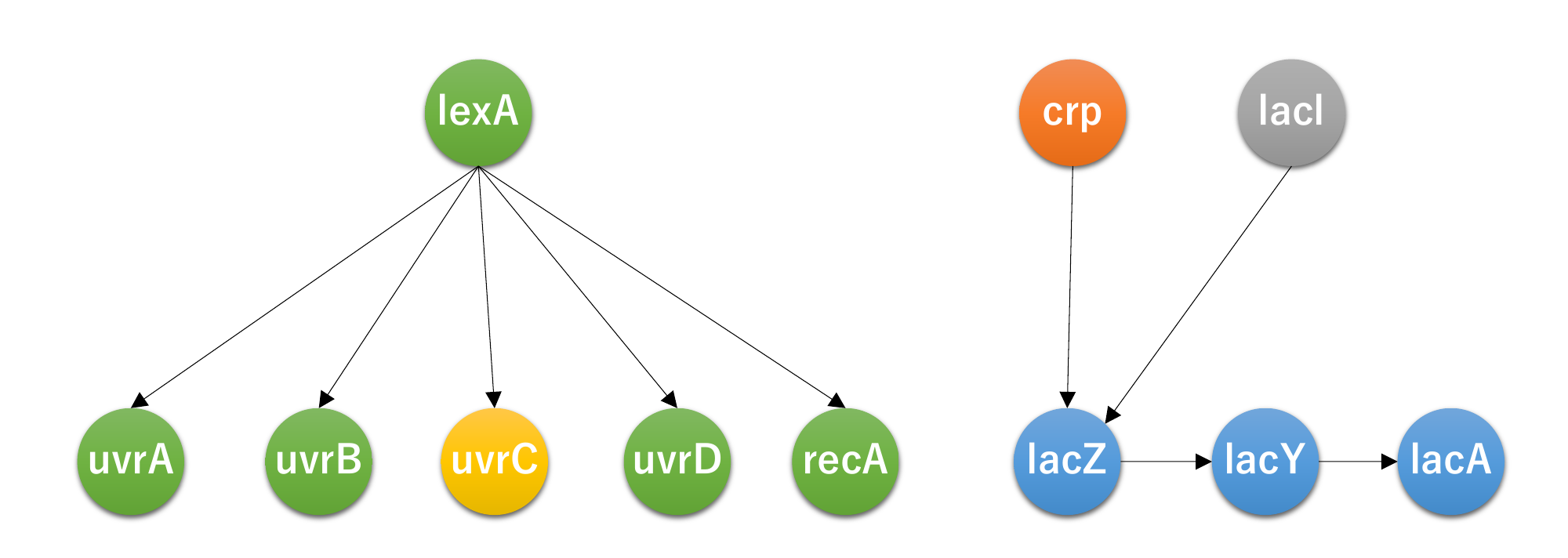} \\
DPVC & \includegraphics[page=6,scale=0.25, clip=true, trim=0cm 4cm 0cm 3cm]{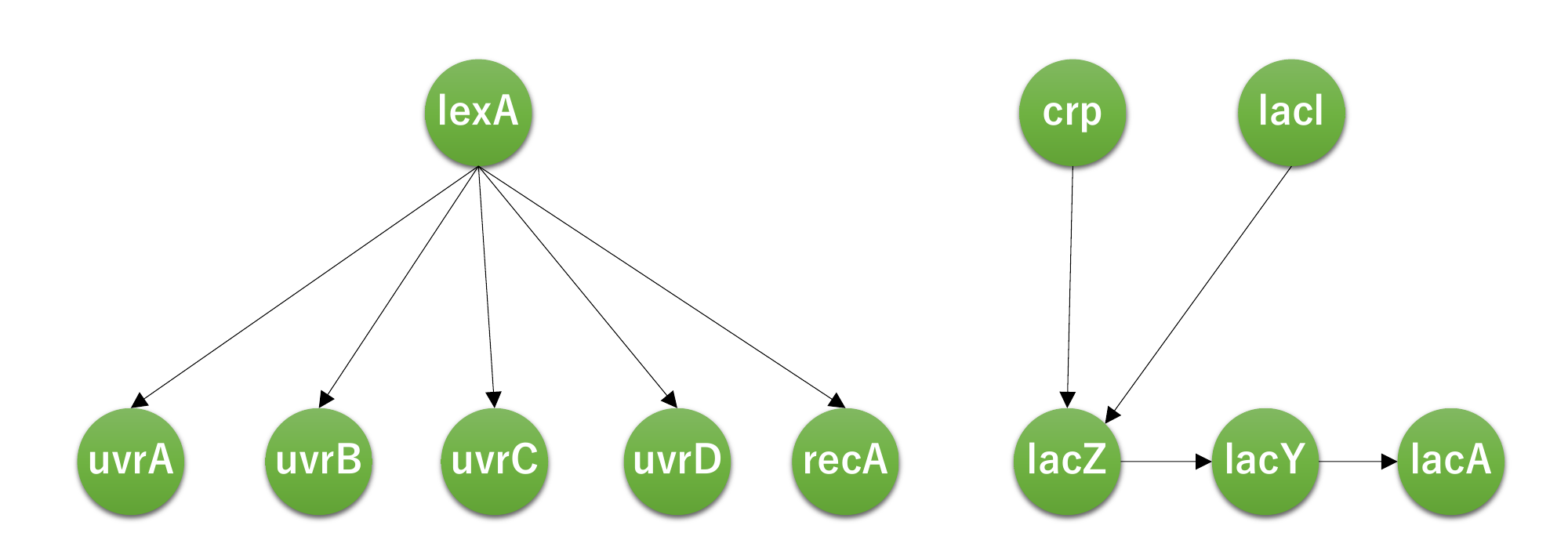} \\
\end{tabular}
\caption{Clusterings of gene regulations network of E. coli. The clustering results are visualized by different colors. Here the size of the restricted hypotheses space $|\mathscr{C}^*|$ found by spectral clustering was 18.}
  \label{fig:gene_regularization_clustering_result}
\end{figure*}

\subsection{Aviation Sensors}

As a third data set, we use the flight aviation dataset from NASA\footnote{\url{https://c3.nasa.gov/dashlink/projects/85/} where we use all records from Tail 687.}.
The data set contains sensor information sampled from airplanes during operation.
We extracted the information of 16 continuous-valued sensors that were recorded for different flights with in total 25032364 samples.

The clustering results are shown in Table \ref{tab:aviationDataResults}.
The data set does not have any ground truth, but the clustering result of our proposed method is reasonable:
Cluster 9 groups sensors that measure or affect altitude\footnote{The elevator position of an airplane influences the altitude, and the static pressure system of an airplane measures the altitude.}, 
Cluster 8 correctly clusters the left and right sensors for measuring the rotation around the axis pointing through the noise of the aircraft, 
 in Cluster 2 all sensors that measure the angle between chord and flight direction are grouped together.
It also appears reasonable that the yellow hydraulic system of the left part of the plane has little direct interaction with the green hydraulic system of the right part (Cluster 1 and Cluster 4).
And the sensor for the rudder, influencing the direction of the plane, is mostly independent of the other sensors (Cluster 5).

In contrast, the clustering selected by the basic inverse Wishart prior, EBIC, and AIC is difficult to interpret.
We note that we did not compare to DPVC, since the large number of samples made the MCMC algorithm of DPVC infeasible. 

 \begin{table*}[h]
  \caption{Evaluation of selected clusterings of the Aviation Sensor Data with 16 variables. Here the size of the restricted hypotheses space $|\mathscr{C}^*|$ found by spectral clustering was 28.}
  \label{tab:aviationDataResults}
 \center
 \small
 \begin{tabular}{ll}
   \toprule
 \multicolumn{2}{c}{Proposed} \\
   \midrule
   \footnotesize Cluster 1 &  \footnotesize BRAKE PRESSURE LH YELLOW \\
      \midrule
\footnotesize Cluster 2 & \footnotesize INDICATED ANGLE OF ATTACK, ANGLE OF ATTACK 2, ANGLE OF ATTACK 1 \\
   \midrule
\footnotesize Cluster 3 & \footnotesize ROLL SPOILER RIGHT \\
   \midrule
\footnotesize Cluster 4 & \footnotesize BRAKE PRESSURE RH GREEN \\
   \midrule
\footnotesize Cluster 5 & \footnotesize RUDDER POSITION \\
   \midrule
\footnotesize Cluster 6 & \footnotesize AILERON POSITION RH, AILERON POSITION LH \\
   \midrule
\footnotesize Cluster 7 & \footnotesize ROLL SPOILER LEFT \\
   \midrule
\footnotesize Cluster 8 & \footnotesize PITCH TRIM POSITION \\
   \midrule
\footnotesize Cluster 9 & \footnotesize STATIC PRESSURE LSP, TOTAL PRESSURE LSP, AVARAGE STATIC PRESSURE LSP, \\
& \footnotesize ELEVATOR POSITION LEFT,ELEVATOR POSITION RIGHT \\
  \midrule 
\noalign{\vskip 4mm}  
 \multicolumn{2}{c}{basic inverse Wishart prior, EBIC ($\gamma \in \{0.0, 0.5, 1.0\}$), AIC} \\
   \midrule
   \footnotesize Cluster 1 & \footnotesize STATIC PRESSURE LSP, INDICATED ANGLE OF ATTACK, TOTAL PRESSURE LSP,\\
  & \footnotesize RUDDER POSITION, AILERON POSITION RH, AVARAGE STATIC PRESSURE LSP, \\
  & \footnotesize ELEVATOR POSITION LEFT, ELEVATOR POSITION RIGHT,  PITCH TRIM POSITION, \\
  & \footnotesize ANGLE OF ATTACK 2, ANGLE OF ATTACK 1, AILERON POSITION LH, ROLL SPOILER LEFT, \\
  & \footnotesize BRAKE PRESSURE LH YELLOW, ROLL SPOILER RIGHT \\
  \midrule  
\footnotesize Cluster 2 & \footnotesize BRAKE PRESSURE RH GREEN \\
  \midrule
  \noalign{\vskip 4mm}
 \multicolumn{2}{c}{Calinski-Harabaz Index} \\
   \midrule
 \footnotesize  Cluster 1 & \footnotesize STATIC PRESSURE LSP, TOTAL PRESSURE LSP, AILERON POSITION RH, \\
 & \footnotesize AVARAGE STATIC PRESSURE LSP, ELEVATOR POSITION LEFT, ELEVATOR POSITION RIGHT, \\
 & \footnotesize BRAKE PRESSURE RH GREEN, AILERON POSITION LH, BRAKE PRESSURE LH YELLOW \\
      \midrule
 \footnotesize  Cluster 2 & \footnotesize INDICATED ANGLE OF ATTACK, ANGLE OF ATTACK 2,  ANGLE OF ATTACK 1\\
      \midrule
\footnotesize Cluster 3 & \footnotesize RUDDER POSITION, PITCH TRIM POSITION, ROLL SPOILER LEFT, ROLL SPOILER RIGHT \\
  \midrule
  \noalign{\vskip 4mm}  
 \multicolumn{2}{c}{CGL (ALC)} \\
   \midrule
\footnotesize Cluster 1 & \footnotesize STATIC PRESSURE LSP, TOTAL PRESSURE LSP, AVARAGE STATIC PRESSURE LSP, \\
& \footnotesize ELEVATOR POSITION LEFT, ELEVATOR POSITION RIGHT, BRAKE PRESSURE LH YELLOW \\
\midrule
\footnotesize Cluster 2 & \footnotesize INDICATED ANGLE OF ATTACK, RUDDER POSITION, AILERON POSITION RH, \\
& \footnotesize PITCH TRIM POSITION, BRAKE PRESSURE RH GREEN, ANGLE OF ATTACK 2,  \\
& \footnotesize ANGLE OF ATTACK 1, AILERON POSITION LH, ROLL SPOILER LEFT, ROLL SPOILER RIGHT \\
\bottomrule
 \end{tabular}
\end{table*}

\section{Discussion and Conclusions} \label{sec:conclusions}
We have introduced a new method for evaluating variable clusterings based on the marginal likelihood of a Bayesian model that takes into account noise on the precision matrix.
Since the calculation of the marginal likelihood is analytically intractable, we proposed two approximations: a variational approximation and an approximation based on MCMC.
Experimentally, we found that the variational approximation is considerably faster than MCMC and also leads to accurate model selections.

We compared our proposed method to several standard model selection criteria. 
In particular, we compared to BIC and extended BIC (EBIC) which are often the method of choice for model selection in Gaussian graphical models.
However, we emphasize that EBIC was designed to handle the situation where $p$ is in the order of $n$, and has not been designed to handle noise.
As a consequence, our experiments showed that in practice its performance depends highly on the choice of the $\gamma$ parameter. 
In contrast, the proposed method, with fixed hyper-parameters, shows better performance on various simulated and real data.
 
We also compared our method to other two previously proposed methods, namely Cluster Graphical Lasso (CGL) \citep{tan2015cluster}, and Dirichlet Process Variable Clustering (DPVC) \citep{palla2012nonparametric} that performs jointly clustering and model selection.
However, it appears that in many situations the model selection algorithm of CGL is not able to detect the true model, even if there is no noise.
On the other hand, the Dirichlet process assumption by DPVC appears to be very restrictive, leading again to many situations where the true model (clustering) is missed.
Overall, our method performs better in terms of selecting the correct clustering on synthetic data with ground truth, and selects meaningful clusters on real data.

The python source code for variable clustering and model selection with the proposed method and all baselines is available at \url{https://github.com/andrade-stats/robustBayesClustering}.

\appendix
\section{Convergence of 3-block ADMM \label{app:convergence3BlockADMM}}
We can write the optimization problem in \eqref{eq:MAP_opt_problem} as 
\begin{align*}
& \text{minimize} \; f_1(X_{\epsilon}) + f_2(X_1, \ldots, X_k) + f_3(Z)  \\
& \text{subject to} \; \\
& \; - X - \beta X_{\epsilon} + Z = 0 \, , \\
& X_{\epsilon}, X_1, \ldots, X_k \succ 0\, ,
\end{align*}
with
\begin{align*}
 & f_1(X_{\epsilon}) :=  trace (A_{\epsilon} X_{\epsilon})  - a_{\epsilon} \cdot \log |X_{\epsilon}| \, ,\\
 & f_2(X_1, \ldots, X_k) :=  \sum_{j = 1}^{k}  \Big( trace (A_j X_j)  - a_j \cdot \log |X_j| \Big) \, ,\\
 & f_3(Z) :=  n \cdot trace (S Z)  - n \cdot \log |Z| \, .
 \end{align*}
 
First note that the functions $f_1, f_2$ and $f_3$ are convex proper closed functions.
Since $X_{\epsilon}, X_1, \ldots, X_k \succ 0$, we have due to the equality constraint that $Z \succ 0$.
Assuming that the global minima is attained, we can assume that $Z \preceq \sigma I$, for some large enough $\sigma > 0$.
As a consequence, we have that $\nabla^2 f_3(Z) = Z^{-1} \otimes Z^{-1} \succeq \sigma^{-2} I$, and therefore $f_3$ is a strongly convex function.
Analogously, we have that $f_1$ and $f_2$ are strongly convex functions, and therefore also coercive.
This allows us to apply Theorem 3.2 in \citep{lin2015global} which guarantees the convergence of the 3-block ADMM.

\section{Derivation of variational approximation} \label{app:derivationVA}

Here, we give more details of the KL-divergence minimization from Section \ref{sec:variationalApproximation}.
Recall, that the remaining parameters $\nu_{g, \epsilon} \in \mathbb{R}$ and $\nu_{g,j} \in \mathbb{R}$ are optimized by minimizing the 
KL-divergence between the the factorized distribution $g$ and the posterior distribution $p(\Sigma_{\epsilon}, \Sigma_{1}, \ldots \Sigma_{k} | \mathbf{x}_1, ..., \mathbf{x}_n, \boldsymbol{\eta}, \mathcal{C})$.  
We have
\begin{align*}
KL(g || p) &= - \int g_{\epsilon}(\Sigma_{\epsilon}) \cdot \prod_{j=1}^{k} g_{j}(\Sigma_{j}) \\
&\quad \log \frac{p(\Sigma_{\epsilon}, \Sigma_{1}, \ldots \Sigma_{k}, \mathbf{x}_1, ..., \mathbf{x}_n | \boldsymbol{\eta}, \mathcal{C})}{ g_{\epsilon}(\Sigma_{\epsilon}) \cdot \prod_{j=1}^{k} g_{j}(\Sigma_{j})} d \Sigma_{\epsilon} d\Sigma \\
&\quad + c \\
&= - \frac{1}{2}\E_{g_J, g_{\epsilon}}[n \cdot \log |(\Sigma^{-1} + \beta \Sigma_{\epsilon}^{-1})|] \\
&\quad - \frac{1}{2}\E_{g_{\epsilon}}[(\nu_{\epsilon} + p + 1) \cdot \log |\Sigma_{\epsilon}^{-1}| \\
&\quad- trace ((\Sigma_{\epsilon,0} + \beta nS) \Sigma_{\epsilon}^{-1}) ] - \text{Entropy}[g_{\epsilon}] \\
&\quad + \sum_{j = 1}^{k} \Big( - \frac{1}{2} \E_{g_j}[(\nu_{j} + p_j + 1) \cdot \log |\Sigma_{j}^{-1}| \\
&\quad- trace ((\Sigma_{j,0} + nS_j) \Sigma_{j}^{-1})] - \text{Entropy}[g_{j}] \Big) + c \\
&= - \frac{1}{2} n \E_{g_J, g_{\epsilon}}[\log |\Sigma^{-1} + \beta \Sigma_{\epsilon}^{-1}|] \\
&\quad + \frac{1}{2} (\nu_{\epsilon} + p + 1) \E_{g_{\epsilon}}[ \log |\Sigma_{\epsilon}| ] \\
&\quad+ \frac{1}{2} trace ((\Sigma_{\epsilon,0} + \beta nS) \E_{g_{\epsilon}}[ \Sigma_{\epsilon}^{-1} ] ) - \text{Entropy}[g_{\epsilon}] \\
&\quad + \frac{1}{2} \sum_{j = 1}^{k}  (\nu_{j} + p_j + 1) \E_{g_j}[ \log |\Sigma_{j}|] \\
&\quad+ \frac{1}{2} \sum_{j = 1}^{k}  trace ((\Sigma_{j,0} + nS_j) \E_{g_j}[\Sigma_{j}^{-1}]) \\
&\quad- \sum_{j = 1}^{k} \text{Entropy}[g_{j}] + c \, ,
\end{align*}
where $c$ is a constant with respect to $g_{\epsilon}$ and $g_j$.  
However, the term $E_{g_J, g_{\epsilon}}[\log |\Sigma^{-1} + \beta \Sigma_{\epsilon}^{-1}|]$ cannot be solved analytically, therefore we need to resort to some sort of approximation. 
Assuming that 
\begin{align*}
E_{g_J, g_{\epsilon}}[\log |\Sigma^{-1} + \beta \Sigma_{\epsilon}^{-1}|] \approx E_{g_J, g_{\epsilon}}[\log |\Sigma^{-1}|] \, , 
\end{align*}
we get 
\begin{align*}
KL(g || p) &\approx  - \frac{1}{2} n \E_{g_J, g_{\epsilon}}[\log |\Sigma^{-1}|] \\
&\quad + \frac{1}{2} (\nu_{\epsilon} + p + 1) \E_{g_{\epsilon}}[ \log |\Sigma_{\epsilon}| ] \\
&\quad+ \frac{1}{2} trace ((\Sigma_{\epsilon,0} + \beta nS) \E_{g_{\epsilon}}[ \Sigma_{\epsilon}^{-1} ] ) - \text{Entropy}[g_{\epsilon}] \\
&\quad + \frac{1}{2} \sum_{j = 1}^{k}  (\nu_{j} + p_j + 1) \E_{g_j}[ \log |\Sigma_{j}|] \\
&\quad+ \frac{1}{2} \sum_{j = 1}^{k}  trace ((\Sigma_{j,0} + nS_j) \E_{g_j}[\Sigma_{j}^{-1}]) \\
&\quad- \sum_{j = 1}^{k} \text{Entropy}[g_{j}] + c \\
&=  - \E_{g_{\epsilon}}[ \log  \Big( |\Sigma_{\epsilon}|^{- \frac{1}{2}  (\nu_{\epsilon} + p + 1)} \\
&\quad e^{- \frac{1}{2} trace ((\Sigma_{\epsilon,0} + \beta nS)  \Sigma_{\epsilon}^{-1} ) } \Big)] \\
&\quad - \text{Entropy}[g_{\epsilon}]  - \sum_{j = 1}^{k}  \E_{g_j}[ \log \Big( |\Sigma_{j}|^{- \frac{1}{2}(\nu_{j} + n + p_j + 1)} \\
&\quad e^{-\frac{1}{2} trace ((\Sigma_{j,0} + nS_j) \Sigma_{j}^{-1})} \Big)] + \text{Entropy}[g_{j}] + c \\
&=  - \E_{g_{\epsilon}}[ \log \text{InvW}(\nu_{\epsilon}, \Sigma_{\epsilon,0} + \beta nS)] \\
&\quad- \text{Entropy}[g_{\epsilon}] \\
&\quad- \sum_{j = 1}^{k}  \E_{g_j}[ \log \text{InvW} (\nu_{j} + n,  \Sigma_{j,0} + nS_j) ] \\
&\quad  + \text{Entropy}[g_{j}] + c'  \\ 
&=  KL(g_{\epsilon} \, || \, \text{InvW}(\nu_{\epsilon}, \Sigma_{\epsilon,0} + \beta nS)) \\
&\quad + \sum_{j = 1}^{k}  KL(g_j \, || \, \text{InvW} (\nu_{j} + n,  \Sigma_{j,0} + nS_j)) \\
&\quad+ c' \, ,
\end{align*}
where we used that $\E_{g_J, g_{\epsilon}}[\log |\Sigma^{-1}|]$ \\
$= - \sum_{j = 1}^{k}  \E_{g_j}[\log |\Sigma_j|]$, and $c'$ is a constant with respect to $g_{\epsilon}$ and $g_j$. 

From the above expression, we see that we can optimize the parameters of $g_{\epsilon}$ and $g_j$ independently from each other.
The optimal parameter $\hat{\nu}_{g, \epsilon}$ for $g_{\epsilon}$ is
\begin{align*}
\hat{\nu}_{g, \epsilon} &= \argmin_{\nu_{g, \epsilon}} KL(g_{\epsilon} \, || \, \text{InvW}(\nu_{\epsilon}, \Sigma_{\epsilon,0} + \beta nS)) \\
&= \argmin_{\nu_{g, \epsilon}} (\nu_{\epsilon} + p + 1) \E_{g_{\epsilon}}[ \log |\Sigma_{\epsilon}| ] \\
&\quad + trace ((\Sigma_{\epsilon,0} + \beta nS) \E_{g_{\epsilon}}[ \Sigma_{\epsilon}^{-1} ] ) - 2 \cdot \text{Entropy}[g_{\epsilon}] \\
&= \argmin_{\nu_{g, \epsilon}}  (\nu_{\epsilon} + p + 1) \Big( - p \log 2 + p \log (\nu_{g, \epsilon} + p + 1) \\
&\quad+ \log |\hat{\Sigma}_{\epsilon}| - \sum_{i=1}^{p} \psi \Big(\frac{\nu_{g, \epsilon} - p + i}{2}\Big) \Big) \\
&\quad+ \frac{\nu_{g, \epsilon}}{\nu_{g, \epsilon} + p + 1}  trace ((\Sigma_{\epsilon,0} + \beta nS) \hat{\Sigma}_{\epsilon}^{-1}) \\
&\quad- 2 \log \Gamma_p(\frac{ \nu_{g, \epsilon}}{2}) - \nu_{g, \epsilon} p - p(p + 1) \log (\nu_{g, \epsilon} + p + 1) \\
&\quad+ (\nu_{g, \epsilon} + p + 1) \sum_{i=1}^{p} \psi \Big(\frac{\nu_{g, \epsilon} - p + i}{2} \Big) \\
&= \argmin_{\nu_{g, \epsilon}} p (\nu_{\epsilon} + p + 1) \log (\nu_{g, \epsilon} + p + 1) \\
&\quad- (\nu_{\epsilon} + p + 1) \sum_{i=1}^{p} \psi \Big(\frac{\nu_{g, \epsilon} - p + i}{2}\Big) \\
&\quad + \frac{\nu_{g, \epsilon}}{\nu_{g, \epsilon} + p + 1}  trace ((\Sigma_{\epsilon,0} + \beta nS) \hat{\Sigma}_{\epsilon}^{-1}) \\
&\quad - 2 \log \Gamma_p(\frac{ \nu_{g, \epsilon}}{2}) - \nu_{g, \epsilon} p - p(p + 1) \log (\nu_{g, \epsilon} + p + 1) \\
&\quad + (\nu_{g, \epsilon} + p + 1) \sum_{i=1}^{p} \psi \Big(\frac{\nu_{g, \epsilon} - p + i}{2} \Big) \\
&= \argmin_{\nu_{g, \epsilon}} \frac{\nu_{g, \epsilon}}{\nu_{g, \epsilon} + p + 1}  trace ((\Sigma_{\epsilon,0} + \beta nS) \hat{\Sigma}_{\epsilon}^{-1}) \\
&\quad- 2 \log \Gamma_p(\frac{ \nu_{g, \epsilon}}{2}) - \nu_{g, \epsilon} p + p \nu_{\epsilon} \log (\nu_{g, \epsilon} + p + 1) \\
&\quad + (\nu_{g, \epsilon} - \nu_{\epsilon}) \sum_{i=1}^{p} \psi \Big(\frac{\nu_{g, \epsilon} - p + i}{2} \Big) \, .
\end{align*}

And analogously, we have 
\begin{align*}
\hat{\nu}_{g, j} &= \argmin_{\nu_{g, j}} \,  \frac{\nu_{g, j}}{\nu_{g, j} + p_j + 1}  trace ((\Sigma_{j,0} + nS_j) \hat{\Sigma}_{j}^{-1}) \\
&\quad - 2 \log \Gamma_{p_j}(\frac{ \nu_{g, j}}{2}) - \nu_{g, j} p_j \\
&\quad + p_j (\nu_{j} + n) \log (\nu_{g, j} + p_j + 1) \\
&\quad + (\nu_{g, j} - \nu_{j} - n) \sum_{i=1}^{p_j} \psi \Big(\frac{\nu_{g, j} - p_j + i}{2} \Big) \, . 
\end{align*}

 \section{Spectral Clustering for variable clustering with the Gaussian graphical model \label{app:spectralClustering}}
 Let $S \in \mathbb{R}^{p \times p}$  denote the sample covariance matrix of the observed variables.
Under the assumption that the observations are drawn i.i.d. from a multivariate normal distribution, with mean $\mathbf{0}$ and precision matrix $X + \beta X_{\epsilon}$, the log-likelihood\footnote{Up to a constant that does not depend on $X$.} of the data is given by
\begin{align*}
& \frac{n}{2} ( \log |X + \beta X_{\epsilon}| - trace ((X + \beta X_{\epsilon}) S ) ) \, ,
\end{align*}
where $n$ is the number of observations.
We assume that $X$ is block sparse, i.e. a permutation matrix $P$ exists such that $P^T X P$ is block diagonal. 
If we knew the number of blocks $m$, then we could estimate the block matrix $X$ (and thus the variable clustering) by the following optimization problem. \\
\begin{flushleft}
Optimization Problem 1:
\end{flushleft}
\begin{ceqn}
\begin{align*}
& \minimize_{X \succ 0} - \log |X + \beta X_{\epsilon}| + trace ((X + \beta X_{\epsilon}) S ) \\
& \text{subject to} \\
& \text{X is block sparse with exactly $m$ blocks} \, ,
\end{align*}
\end{ceqn}
where $\beta X_{\epsilon}$ is assumed to be a constant matrix with small entries.
We claim that this can be reformulated, for any $q > 0$, as following. \\
\begin{flushleft}
Optimization Problem 2:
\end{flushleft}
\begin{ceqn}
\begin{align*}
& \minimize_{X \succ 0} - \log |X + \beta X_{\epsilon}| + trace ((X + \beta X_{\epsilon}) S ) \\
& \text{subject to} \\
& L_{ii} = \sum_{k \neq i} |X_{ik}|^q \, ,\\
& L_{ij} =  - |X_{ij}|^q  \; \;  \text{for } i \neq j, \\
& rank(L) = p - m \, .
 \end{align*}
\end{ceqn}
\begin{proposition}
Optimization problem 1 and 2 have the same solution. 
Moreover, the $m$ dimensional  null space of $L$ can be chosen such that each basis vector is the indicator vector for one variable block of $X$. 
\end{proposition}
\begin{proof}
First let us define the matrix $\tilde{X}$, by $\tilde{X}_{ij} := |X_{ij}|^q$.
Then clearly, iff $X$ is block sparse with $m$ blocks, so is $\tilde{X}$.
Furthermore, $\tilde{X}_{ij} \geq 0$, and $L$ is the unnormalized Laplacian as defined in \citep{von2007tutorial}.
We can therefore apply Proposition (2) of \citep{von2007tutorial}, to find that the dimension of the eigenspace of L corresponding to eigenvalue 0, is exactly the number of blocks in $\tilde{X}$.
Also from Proposition (2) of \citep{von2007tutorial} it follows that each such eigenvector $\mathbf{e}_k \in \mathbb{R}^{p}$ can be chosen such that it indicates the variables belonging to the same block, i.e. $\mathbf{e}_k(i) \neq 0$, iff variable i belongs to block k. \qed
\end{proof}

Using the nuclear norm as a convex relaxation for the rank constraint, we have
\begin{align*}
& \minimize_{X \succeq 0} - \log |X + \beta X_{\epsilon}| + trace ((X + \beta X_{\epsilon}) S ) + \lambda_m ||L||_*  \\
& \text{subject to}  \\
& L_{ii} = \sum_{k \neq i} |X_{ik}|^q \, , \\
& L_{ij} =  - |X_{ij}|^q  \; \;  \text{for } i \neq j \, .
\end{align*}
with an appropriately chosen $\lambda_m$. By the definition of $L$, we have that $L$ is positive semi-definite, and therefore $||L||_* = trace(L)$.
As a consequence, we can rewrite the above problem as 
\begin{align*}
& X^{*} := \argmin_{X \succeq 0} - \log |X + \beta X_{\epsilon}| + trace ((X + \beta X_{\epsilon}) S ) \\
& \quad \quad \quad \quad \quad \quad + \lambda_m \sum_{i \neq j} |X_{ij}|^q \, .
\end{align*}
Finally, for the purpose of learning the Laplacian $L$, we ignore the term $\beta X_{\epsilon}$ and set it to zero. 
This will necessarily lead to an estimate of $X^{*}$ that is not a clean block matrix, but has small non-zero entries between blocks.
Nevertheless, spectral clustering is known to be robust to such violations \citep{ng2002spectral}.
This leads to Algorithm \ref{alg:spectralVarClustering} in Section \ref{sec:restrictingHypothesesSpace}.

\bibliographystyle{plainnat}
\bibliography{paperReferences_new.bib}

\end{document}